%% file: ExtremalArxiv.tex
 \documentclass[conference, onecolumn,pagenumbers]{IEEEtran}

\IEEEoverridecommandlockouts
\usepackage{graphicx}
\usepackage{cite}
\usepackage{color}
\usepackage{amsfonts}
\usepackage{amsthm}
\usepackage{url}
\usepackage{amsmath}
\usepackage{amssymb}
\usepackage{subfig}
\usepackage{fixltx2e}

\usepackage{pseudocode}

\newtheorem{theorem}{Theorem}
\newtheorem{lemma}{Lemma}
\newtheorem{remark}{Remark}
\newtheorem{definition}{Definition}
\newtheorem{corollary}{Corollary}

\newtheorem{proposition}{Proposition}

\newcommand{\EE}{\mathbb{E}}

\newcommand{\matr}[1]{{\underline{#1}}}

\newcommand{\xh}[1]{\hat{X}_{#1}}
\newcommand{\yh}[1]{\hat{Y}_{#1}}
\newcommand{\uh}{\hat{U}}
\newcommand{\vh}{\hat{V}}

\newcommand{\vecY}{\mathbf{Y}}
\newcommand{\vecX}{\mathbf{X}}

\newcommand{\Tr}{\operatorname{Tr}}

\graphicspath{{figs/}}

\begin{document}

\title{An Extremal Inequality for Long Markov Chains
\thanks{This work was supported by the NSF Center for Science of Information
under grant agreement CCF-0939370.}
}

\author{
\IEEEauthorblockN{Thomas~A.~Courtade}
\IEEEauthorblockA{Department of Electrical Engineering and Computer Sciences\\University of California, Berkeley
              }
\and
\authorblockN{Jiantao Jiao}
\authorblockA{Department of Electrical Engineering\\
Stanford University}
}
\date{\today}

\vspace{-10pt}

\maketitle

\thispagestyle{plain}
\pagestyle{plain}

\vspace{-15pt}

\begin{abstract}
Let $\vecX,\vecY$ be jointly Gaussian vectors, and consider random variables $U,V$ that satisfy the Markov constraint $U-\vecX-\vecY-V$.  We prove an extremal inequality relating the mutual informations between all ${4 \choose 2}$ pairs of random variables from the set $(U,\vecX,\vecY,V)$.  As a first application, we show that the rate region for the two-encoder quadratic Gaussian source coding problem follows as an immediate corollary of the the extremal inequality.  In a second application, we establish the rate region for a vector-Gaussian  source coding problem where L\"owner-John ellipsoids are approximated based on rate-constrained descriptions of the data.
\end{abstract}

\section{Introduction}
In this paper, we prove the following extremal result, which resembles an entropy power inequality:
\begin{theorem}\label{thm:vector}
For $n\times n$ positive definite matrices $\Sigma_X, \Sigma_Z$, let $\mathbf{X}\sim N(\mu_X,\Sigma_X)$ and  $\mathbf{Z}\sim N(\mu_Z,\Sigma_Z)$ be independent $n$-dimensional Gaussian vectors, and define $\mathbf{Y} = \mathbf{X} + \mathbf{Z}$.   For any $U,V$ such that $U-\mathbf{X}-\mathbf{Y}-V$ form a Markov chain, the following inequality holds:
\begin{align}
2^{-\frac{2}{n} (I(\mathbf{Y};U)+ I(\mathbf{X};V))} &\geq \frac{| \Sigma_X |^{1/n}}{| \Sigma_X  + \Sigma_Z |^{1/n}}  ~2^{-\frac{2}{n} (I(\mathbf{X};U)+I(\mathbf{Y};V))} + 2^{-\frac{2}{n}( I(\mathbf{X};\mathbf{Y})+I(U;V))}.\label{introVec}%
\end{align}
\end{theorem}
\noindent In the simplest case, where $\mathbf{Y} = \rho \mathbf{X} + \mathbf{Z}$,  $\Sigma_X = I_n$ and $\Sigma_Z = (1-\rho^2) I_n$, Theorem  \ref{thm:vector} implies
\begin{align}
2^{-\frac{2}{n} (I(\mathbf{Y};U)+ I(\mathbf{X};V))}  &\geq (1-\rho^2)2^{-\frac{2}{n} I(V;U)} + \rho^2 2^{-\frac{2}{n} (I(\mathbf{X};U)+I(\mathbf{Y};V))}. \label{introWhite}%
\end{align}
If $V$ is degenerate, \eqref{introWhite} further simplifies to an inequality shown by Oohama in \cite{Oohama1997}, which proved to be instrumental in establishing the rate-distortion region for the one-helper quadratic Gaussian source coding problem.  Together with Oohama's work, the sum-rate constraint established by Wagner \emph{et al}. in their \emph{tour de force}  \cite{WagnerRateRegion2008} completely characterized the rate-distortion region for  the two-encoder quadratic Gaussian source coding problem.  It turns out that the sum-rate constraint of Wagner \emph{et al}. can be recovered as an immediate corollary to \eqref{introWhite}, thus unifying the works of Oohama and Wagner \emph{et al}. under a common inequality.  The entire argument is given as follows. 

\subsection{Recovery of the scalar-Gaussian sum-rate constraint}
Using the Markov relationship $U-\mathbf{X}-\mathbf{Y}-V$, we can rearrange the exponents in \eqref{introWhite} to obtain the equivalent inequality 
\begin{align}
2^{-\frac{2}{n} (I(\mathbf{X};U,V) +I(\mathbf{Y};U,V))} \geq 2^{-\frac{2}{n}I(\mathbf{X},\mathbf{Y};U,V) } \left( 1-\rho^2 + \rho^2 2^{-\frac{2}{n}I(\mathbf{X},\mathbf{Y};U,V) }\right) . \label{LRHSmonotone}
\end{align}
The left- and right-hand sides of \eqref{LRHSmonotone} are monotone decreasing in $\frac{1}{n}(I(\mathbf{X};U,V) +I(\mathbf{Y};U,V))$ and $\frac{1}{n}I(\mathbf{X},\mathbf{Y};U,V)$, respectively.  Therefore, if 
\begin{align}
&\frac{1}{n} (I(\mathbf{X};U,V) +I( \mathbf{Y};U,V) ) \geq \frac{1}{2} \log \frac{1}{D} \mbox{~~~and~~~}\frac{1}{n}I(\mathbf{X},\mathbf{Y};U,V) \leq R \label{eqnRD}
\end{align}
for some pair $(R,D)$, then we have $D \geq 2^{-2 R } \left( 1-\rho^2 + \rho^2 2^{-2 R }\right)$,
 which is a quadratic inequality with respect to the term $2^{-2 R }$.  This is easily solved using the quadratic formula to obtain:
\begin{align}
2^{-2R} \leq \frac{2D}{(1-\rho^2)\beta(D)} \quad \Rightarrow \quad R \geq \frac{1}{2}\log \frac{(1-\rho^2)\beta(D)}{2D}, \label{RDineq}
\end{align}
where $\beta(D)\triangleq 1 + \sqrt{1 + \frac{4\rho^2 D}{(1-\rho^2)^2}}$.  Note that Jensen's inequality and the maximum-entropy property of Gaussians imply 
\begin{align}
\frac{1}{n} (I(\mathbf{X};U,V) +I( \mathbf{Y};U,V) ) \geq \frac{1}{2} \log \frac{1}{\mathsf{mmse}(\mathbf{X}|U,V) \mathsf{mmse}(\mathbf{Y}|U,V) },\label{mmseEq}
\end{align}
where $\mathsf{mmse}(\mathbf{X}|U,V) \triangleq \frac{1}{n} \| \mathbf{X} - \mathbb{E}[\mathbf{X}|U,V]  \|^2$, and $\mathsf{mmse}(\mathbf{Y}|U,V)$ is defined similarly.  Put $U = f_x(\mathbf{X})$ and $V = f_y(\mathbf{Y})$, where $f_x : \mathbb{R}^n \rightarrow [1:2^{nR_x}]$ and $f_y : \mathbb{R}^n \rightarrow [1:2^{nR_y}]$.  Supposing  $\mathsf{mmse}(\mathbf{X}|U,V)\leq d_x$ and $\mathsf{mmse}(\mathbf{Y}|U,V)\leq d_y$,  inequalities \eqref{eqnRD}-\eqref{mmseEq} together imply
\begin{align}
R_x + R_y \geq \frac{1}{2}\log \frac{(1-\rho^2)\beta(d_x d_y)}{2d_x d_y},\label{recoveredSumRate}
\end{align}
which is precisely the sum-rate constraint for the two-encoder quadratic Gaussian source coding problem.

\subsection{Distributed compression of minimal-volume ellipsoids}

Above, recovery of the quadratic Gaussian sum-rate constraint \eqref{recoveredSumRate} demonstrated the utility of Theorem \ref{thm:vector} in proving nontrivial results.  Now, we consider a new problem which, the the authors' knowledge, is not a consequence of known results in the open literature.  In particular, we study the problem of compressing ellipsoids that cover a set of points which, subject to rate constraints, have approximately minimal volume.  Such ellipsoids are similar to L\"owner-John ellipsoids, which are defined as the (unique) ellipsoid of minimal volume that covers a finite set of points \cite{henk2012lowner}.  These minimum-volume ellipsoids and their approximations play a prominent role in the fields of optimization, data analysis, and computational geometry (e.g., \cite{boyd2004convex}). 

To begin, we recall that an $n$-dimensional ellipsoid $\mathcal{E}$ can be parameterized by a positive semidefinite matrix $A\in \mathbb{R}^{n\times n}$ and a vector $b\in \mathbb{R}^n$ as follows:
\begin{align}
\mathcal{E} = \mathcal{E}(A,b) = \left\{ x\in \mathbb{R}^n : \|Ax - b\|\leq 1 \right\}.\label{ellParam}
\end{align}
The volume of $\mathcal{E}(A,b)$ is related to the determinant of $A$ by
\begin{align}
\operatorname{vol}\left(\mathcal{E}(A,b)\right) = \frac{c_n}{|A|},
\end{align}
where $c_n \sim \frac{1}{\sqrt{n\pi}}\left(\frac{2\pi e}{n}\right)^{n/2}$ is the volume of the $n$-dimensional unit ball.

Fix $\rho\in(0,1)$, and let   
 $\{\Sigma_n : n\geq 1\}$ be a sequence of positive definite $n\times n$ matrices. Suppose $(\vecX_1, \vecY_1), \dots, (\vecX_k, \vecY_k)$ are $k$ independent pairs of jointly Gaussian vectors, each equal in distribution to $(\vecX,\vecY)$, where $\mathbb{E}[\vecX \vecX^T] = \mathbb{E}[\vecY \vecY^T]=\Sigma_n$, and $\mathbb{E}[\vecX \vecY^T] =\rho\Sigma_n$. 

 A $(n,R_x,R_y,\nu_x,\nu_y,k,\Sigma_n, \epsilon)$-code consists of encoding functions 
\begin{align}
f_x &: \mathbb{R}^{k n} \rightarrow \{1,2,\dots,2^{kn R_x }\}\\
f_y &: \mathbb{R}^{k n} \rightarrow \{1,2,\dots,2^{kn R_y }\}
\end{align}
and a decoding function
\begin{align}
\psi : \left(f_x(\vecX_1, \dots, \vecX_k), f_y(\vecY_1, \dots, \vecY_k)\right) \mapsto (A_x, A_y, b_x, b_y)
\end{align}
such that 
\begin{align}
\max_{1\leq i\leq k} \Pr\left\{ \vecX_i \notin \mathcal{E}(A_x,b_x) \right\} < \epsilon \mbox{~~and~~} 
\max_{1\leq i\leq k} \Pr\left\{ \vecY_i \notin \mathcal{E}(A_y,b_y) \right\} < \epsilon,
\end{align}
and 
\begin{align}
\Big( \operatorname{vol}(\mathcal{E}(A_x,b_x) ) \Big)^{1/n} &\leq (1+\epsilon) {c_n^{1/n}}{\sqrt{{n \nu_x}  |\Sigma_n|^{1/n}}} \label{vol1}\\
\Big( \operatorname{vol}(\mathcal{E}(A_y,b_y) ) \Big)^{1/n} &\leq (1+\epsilon) {c_n^{1/n}}{\sqrt{{n \nu_y}  |\Sigma_n|^{1/n}}}.\label{vol2}
\end{align}%

We remark that $\sqrt{n}c_n^{1/n}\rightarrow \sqrt{2 \pi e}$  as $n\rightarrow \infty$ by Stirling's approximation, which explains the normalization factor of $\sqrt{n}$ in the volume constraint.  In particular, \eqref{vol1}-\eqref{vol2} can be replaced with
\begin{align}
\Big( \operatorname{vol}(\mathcal{E}(A_x,b_x) ) \Big)^{1/n} &\leq (1+\epsilon) {\sqrt{{(2\pi e) \nu_x}  |\Sigma_n|^{1/n}}} \label{vol1a}\\
\Big( \operatorname{vol}(\mathcal{E}(A_y,b_y) ) \Big)^{1/n} &\leq (1+\epsilon) {\sqrt{{(2\pi e) \nu_y}  |\Sigma_n|^{1/n}}}.\label{vol2a}
\end{align}

\begin{definition}
For a sequence $\{\Sigma_n : n\geq 1\}$ of positive definite $n\times n$ matrices, a tuple $(R_x,R_y,\nu_x,\nu_y,k)$ is $\{\Sigma_n : n\geq 1\}$-achievable if there exists a sequence of $(n,R_x,R_y,\nu_x,\nu_y,k,\Sigma_n,\epsilon_n)$ codes satisfying $\epsilon_n\rightarrow 0$ as $n\rightarrow \infty$.
\end{definition}

If $(R_x,R_y,\nu_x,\nu_y,k)$ is a Pareto-optimal $\{\Sigma_n : n\geq 1\}$-achievable point, the corresponding ellipsoids $\mathcal{E}(A_x,b_x) ,\mathcal{E}(A_y,b_y)$ can be viewed as the best approximations to  L\"owner-John ellipsoids subject to rate-constrained descriptions of the data.  That is, the two ellipsoids cover the $k$ points observed at their respective encoders, and are (essentially) the minimum-volume such ellipsoids that can be computed from rate-constrained descriptions of the data. The general problem setup is illustrated in Figure \ref{fig:ellipse}.

\begin{figure}
\def\svgwidth{\textwidth}
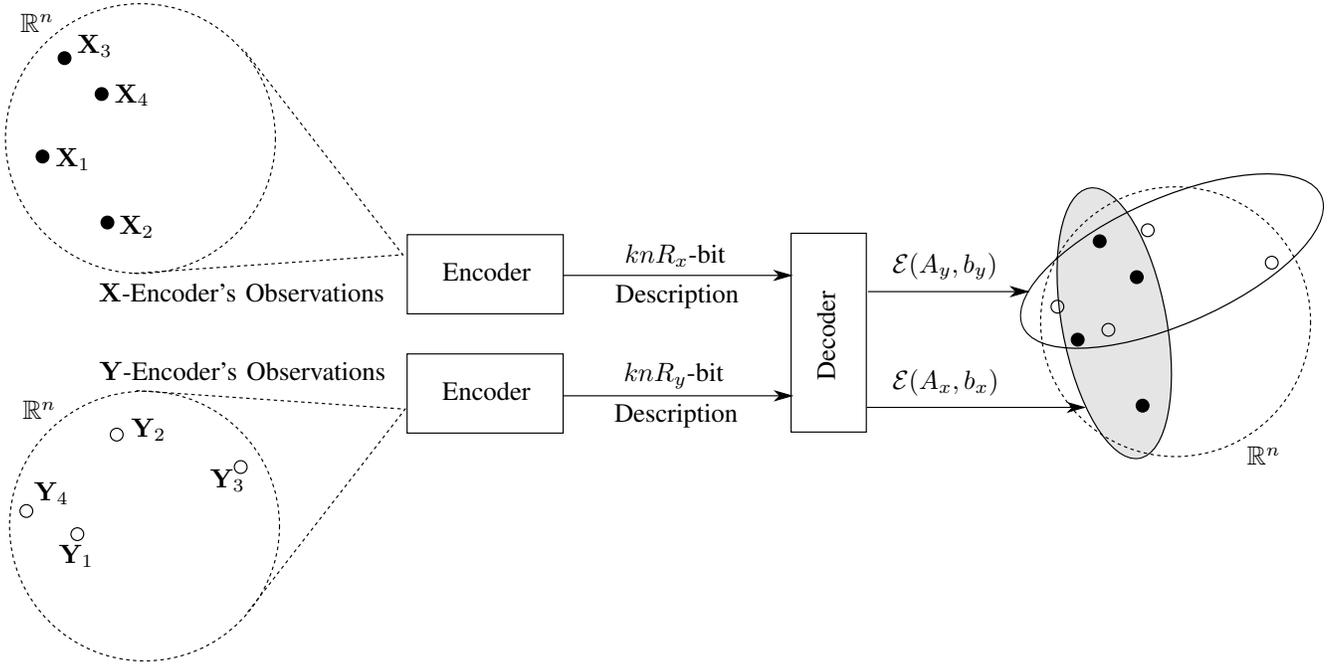
\caption{Computation of covering-ellipsoids from compressed descriptions of the observed data points ($k=4$).  Note the decoder only computes the ellipsoids $\mathcal{E}(A_x,b_x),\mathcal{E}(A_y,b_y)$.  The data points at the output of the decoder are only shown for reference.}\label{fig:ellipse}
\end{figure}

\begin{theorem}\label{thm:Ellipsoid}
For any sequence $\{\Sigma_n : n\geq 1\}$ of positive definite $n\times n$ matrices, a tuple $(R_x,R_y,\nu_x,\nu_y,k)$ is $\{\Sigma_n : n\geq 1\}$-achievable if and only if
\begin{align}
R_x &\geq\frac{1}{2}\log \left[\frac{1}{\nu_x}\left(1-\rho^2 + \rho^2 2^{-2 R_y} \right) \right]\label{RxConstratint}\\%
R_y &\geq\frac{1}{2}\log \left[\frac{1}{\nu_y}\left(1-\rho^2 + \rho^2 2^{-2 R_x} \right) \right]\label{RyConstratint}\\
R_x + R_y &\geq \frac{1}{2}\log \frac{(1-\rho^2)\beta(\nu_x \nu_y)}{2\nu_x \nu_y},\label{RxyConstratint}
\end{align}
where $\beta(z)\triangleq 1 + \sqrt{1 + \frac{4\rho^2 z}{(1-\rho^2)^2}}$.
\end{theorem}

\begin{remark}
As we will see, the direct part of Theorem \ref{thm:Ellipsoid} follows from an application of the achievability scheme for the two-encoder quadratic Gaussian source coding problem.  However, the converse result does not appear to be a similar consequence since the matrices $A_x,A_y$ describing the principal axes of the ellipsoids are allowed to depend on the source realizations.  %
Nonetheless, with Theorem \ref{thm:vector} at our disposal, the proof of the converse is fairly routine.  
\end{remark}
Since the primary focus of this paper is on the extremal inequality \eqref{introVec}, we defer the proof of Theorem \ref{thm:Ellipsoid} until Appendix \ref{app:ellipsoidProof}.  The remainder of this paper is divided into two parts: a treatment of the scalar version of \eqref{introVec} is given in Section \ref{sec:scalar}, and the  vector generalization  is considered in Section \ref{sec:vector}.  Closing remarks are provided in Section \ref{sec:conclusion}.

\section{Scalar Setting}\label{sec:scalar}
We begin the journey toward our main result by studying the scalar version of Theorem \ref{thm:vector}.  Most of our effort will carry over to the vector setting, but the notation in the scalar case is less cumbersome.  Therefore, for the remainder of this  section, we will assume that $X,Y$ are jointly Gaussian, each with unit variance and correlation $\rho$. Our main result in this section is the following rearrangement of \eqref{introVec}.
\begin{theorem}\label{thm:scalar}
Suppose $X,Y$ are jointly Gaussian, each with unit variance and correlation $\rho$.  Then, for any $U,V$ satisfying $U-X-Y-V$, the following inequality holds:
\begin{align}
2^{-2I(Y;U)}  2^{-2I(X;V|U)} \geq (1-\rho^2) + \rho^2 2^{-2I(X;U)} 2^{-2I(Y;V|U)}  .\label{main}
\end{align}
\end{theorem}

\subsection{Proof of Theorem \ref{thm:scalar}}
Instead of working directly with inequality \eqref{main}, it will be convenient to consider a dual form.  To this end, for $\lambda\geq 0$, define 
\begin{align}
F(\lambda)\triangleq \inf_{U,V: U-X-Y-V}\Big\{ I(X;U)-\lambda I(Y;U) + I(Y;V|U)-\lambda I(X;V|U)\Big\}. \label{funL}
\end{align}
The remainder of this section is devoted to characterizing the function $F(\lambda)$.  We remark that that the infimum in \eqref{funL} is attained for any $\lambda$.  The proof of this is routine, and deferred to Appendix \ref{app:InfimaObtained}.  The bulk of the work ahead is devoted to establishing the existence of valid minimizers $U,V$ for which $X|\{U = u\}$ is normal for almost every $u$.

To accomplish this, we now describe a simple construction that will be used throughout much of the sequel.  This construction was first introduced for proving extremal inequalities in \cite{GengNairPrivateMessages2014}.  Suppose $U,X,Y,V$ satisfy the Markov relationship $U-X-Y-V$, and consider two independent copies of $U,X,Y,V$, which will be denoted by the same variables with subscripts $1$ and $2$.  Define
\begin{align}
&\xh1 = \frac{X_1 + X_2}{\sqrt{2}} &\xh2 = \frac{X_1 - X_2}{\sqrt{2}}.
\end{align}
In a similar manner, define $\yh1, \yh2$.  Note that $(\xh1,\xh2,\yh1,\yh2)$ and $(X_1,X_2,Y_1,Y_2)$ are equal in distribution.  Let $g:\mathbb{R}^2\rightarrow \mathbb{R}$ be a one-to-one measurable transformation\footnote{Every uncountable Polish space is Borel isomorphic to $\mathbb{R}$ \cite{srivastava1998course}.}. Define $\uh = g(U_1,U_2)$ and $\vh = g(V_1,V_2)$.

\begin{lemma} \label{lem:candidateMinimizers}
If  $U,X,Y,V$ minimize the functional \eqref{funL}, and $\xh1,\xh2,\yh1,\yh2,\uh,\vh$ are constructed as  above, then 
\begin{enumerate}
\item For almost every $y$, $\uh,\xh1,\yh1,\vh$ conditioned on $\{\yh2=y\}$ is a valid minimizer of \eqref{funL}.
\item For almost every $y$,  $\uh,\xh2,\yh2,\vh$ conditioned on $\{\yh1=y\}$ is a valid minimizer of \eqref{funL}.
\end{enumerate}
\end{lemma}
\begin{proof}
Let $\phi_1$ be such that $\xh1,\yh1$ are independent of $\phi_1$ and $\uh-\xh1-\yh1 - \vh$ conditioned on $\phi_1$. Valid assignments of $\phi_1$ include any nonempty subset of $\xh2, \yh2$.  Let $\phi_2$ be defined similarly. 

Now, observe that we can write:
\begin{align}
2 I(X;U) &= I(\xh1,\xh2;\uh)  = I(\xh1;\uh) + I(\xh2;\uh | \xh1) \\
&= I(\xh1;\uh) + I(\xh2;\uh , \xh1) \\
&=I(\xh1;\uh) + I(\xh2;\uh) + I(\xh2; \xh1|\uh) \\
&=I(\xh1;\uh|\phi_1) + I(\xh2;\uh|\phi_2) - I(\xh1;\phi_1|\uh)  - I(\xh2;\phi_2|\uh) + I(\xh1; \xh2|\uh).
\end{align}
Similarly,
\begin{align}
2 I(Y;U) 
&=I(\yh1;\uh|\phi_1) + I(\yh2;\uh|\phi_2) - I(\yh1;\phi_1|\uh)  - I(\yh2;\phi_2|\uh) + I(\yh1; \yh2|\uh)\\
\end{align}
Also, we have
\begin{align}
2I(Y;V|U) &= I(\yh1,\yh2;\vh|\uh) =  I(\yh1;\vh|\uh) +  I(\yh2;\vh |\uh)  -I(\yh2;\yh1|\uh) + I(\yh2;\yh1|\uh,\vh) \\
&= I(\yh1; \phi_1,\vh|\uh) +  I(\yh2;\phi_2,\vh |\uh) - I(\yh1;\phi_1|\uh,\vh)  - I(\yh2;\phi_2|\uh,\vh) \\
&\quad-I(\yh2;\yh1|\uh) + I(\yh2;\yh1|\uh,\vh) \notag\\
&= I(\yh1; \vh|\uh,\phi_1) +  I(\yh2;\vh |\uh,\phi_2) \\
&\quad- I(\yh1;\phi_1|\uh,\vh)  - I(\yh2;\phi_2|\uh,\vh) + I(\yh2;\yh1|\uh,\vh) \notag\\
&\quad  + I(\yh1; \phi_1|\uh) + I(\yh2; \phi_2|\uh) -I(\yh1;\yh2|\uh).\notag
\end{align}
And, similarly,
\begin{align}
2I(X;V|U) &= I(\xh1; \vh|\uh,\phi_1) +  I(\xh2;\vh |\uh,\phi_2) \\
&\quad- I(\xh1;\phi_1|\uh,\vh)  - I(\xh2;\phi_2|\uh,\vh) + I(\xh2;\xh1|\uh,\vh) \notag\\
&\quad  + I(\xh1; \phi_1|\uh) + I(\xh2; \phi_2|\uh) -I(\xh1;\xh2|\uh).\notag
\end{align}

Assume $U,V$ minimize the functional \eqref{funL} subject to the Markov constraint $U-X-Y-V$, the existence of such $U,V$ was established in Lemma \ref{lem:InfAttained}.  Then, combining above, we have
\begin{align}
2F(\lambda) =&  I(\xh1;\uh|\phi_1) + I(\xh2;\uh|\phi_2) - I(\xh1;\phi_1|\uh)  - I(\xh2;\phi_2|\uh) + I(\xh1; \xh2|\uh)\\
& -\lambda\left( I(\yh1;\uh|\phi_1) + I(\yh2;\uh|\phi_2) - I(\yh1;\phi_1|\uh)  - I(\yh2;\phi_2|\uh) + I(\yh1; \yh2|\uh) \right) \notag\\
& + I(\yh1; \vh|\uh,\phi_1) +  I(\yh2;\vh |\uh,\phi_2) \notag\\
&\quad- I(\yh1;\phi_1|\uh,\vh)  - I(\yh2;\phi_2|\uh,\vh) + I(\yh2;\yh1|\uh,\vh) \notag\\
&\quad  + I(\yh1; \phi_1|\uh) + I(\yh2; \phi_2|\uh) -I(\yh1;\yh2|\uh)\notag\\
& -\lambda \Big( I(\xh1; \vh|\uh,\phi_1) +  I(\xh2;\vh |\uh,\phi_2) \notag\\
&\quad- I(\xh1;\phi_1|\uh,\vh)  - I(\xh2;\phi_2|\uh,\vh) + I(\xh2;\xh1|\uh,\vh) \notag\\
&\quad  + I(\xh1; \phi_1|\uh) + I(\xh2; \phi_2|\uh) -I(\xh1;\xh2|\uh)\Big) \notag\\
 =& I(\xh1;\uh|\phi_1) - \lambda I(\yh1;\uh|\phi_1) + I(\yh1;\vh |\uh,\phi_1)  - \lambda  I(\xh1;\vh |\uh,\phi_1) \\
& + I(\xh2;\uh|\phi_2)  -\lambda I(\yh2;\uh|\phi_2) + I(\yh2;\vh |\uh,\phi_2) - \lambda I(\xh2;\vh |\uh,\phi_2)\notag\\
& -(\lambda+1) I(\xh1;\phi_1|\uh)  -(\lambda+1) I(\xh2;\phi_2|\uh) + (\lambda+1) I(\xh1; \xh2|\uh) \notag\\
& +(\lambda+1) I(\yh1;\phi_1|\uh)  +(\lambda+1) I(\yh2;\phi_2|\uh) - (\lambda+1) I(\yh1; \yh2|\uh) \notag\\
& - I(\yh1;\phi_1|\uh,\vh)  - I(\yh2;\phi_2|\uh,\vh) + I(\yh2;\yh1|\uh,\vh) \notag\\
& +\lambda I(\xh1;\phi_1|\uh,\vh)  +\lambda I(\xh2;\phi_2|\uh,\vh) - \lambda I(\xh2;\xh1|\uh,\vh) \notag\\
 \geq& 2 F(\lambda) \\
& -(\lambda+1) I(\xh1;\phi_1|\uh)  -(\lambda+1) I(\xh2;\phi_2|\uh) + (\lambda+1) I(\xh1; \xh2|\uh) \notag\\
& +(\lambda+1) I(\yh1;\phi_1|\uh)  +(\lambda+1) I(\yh2;\phi_2|\uh) - (\lambda+1) I(\yh1; \yh2|\uh) \notag\\
& - I(\yh1;\phi_1|\uh,\vh)  - I(\yh2;\phi_2|\uh,\vh) + I(\yh2;\yh1|\uh,\vh) \notag\\
& +\lambda I(\xh1;\phi_1|\uh,\vh)  +\lambda I(\xh2;\phi_2|\uh,\vh) - \lambda I(\xh2;\xh1|\uh,\vh) .\notag
\end{align}
The last inequality follows since $\uh-\xh1-\yh1 - \vh$ conditioned on $\phi_1$ is a candidate minimizer of the functional, and same for $\uh-\xh2-\yh2 - \vh$ conditioned on $\phi_2$.  Hence, we can conclude that the following must hold
\begin{align}
& (\lambda+1) I(\xh1;\phi_1|\uh)  +(\lambda+1) I(\xh2;\phi_2|\uh) - (\lambda+1) I(\xh1; \xh2|\uh) \notag\\
& + I(\yh1;\phi_1|\uh,\vh)  + I(\yh2;\phi_2|\uh,\vh) - I(\yh2;\yh1|\uh,\vh) \notag\\
\geq&(\lambda+1) I(\yh1;\phi_1|\uh)  +(\lambda+1) I(\yh2;\phi_2|\uh) - (\lambda+1) I(\yh1; \yh2|\uh) \label{mainIneq}\\
& +\lambda I(\xh1;\phi_1|\uh,\vh)  +\lambda I(\xh2;\phi_2|\uh,\vh) - \lambda I(\xh2;\xh1|\uh,\vh) .\notag
\end{align}

Now, set $\phi_1 = \xh2, \phi_2 = \yh1$.  The LHS of \eqref{mainIneq} is given by
\begin{align}
& (\lambda+1) I(\xh1;\phi_1|\uh)  +(\lambda+1) I(\xh2;\phi_2|\uh) - (\lambda+1) I(\xh1; \xh2|\uh) \notag\\
& + I(\yh1;\phi_1|\uh,\vh)  + I(\yh2;\phi_2|\uh,\vh) - I(\yh2;\yh1|\uh,\vh) \notag\\
=&  (\lambda+1) I(\xh1;\xh2|\uh)  +(\lambda+1) I(\xh2;\yh1|\uh) - (\lambda+1) I(\xh1; \xh2|\uh) \\
& + I(\yh1;\xh2|\uh,\vh)  + I(\yh2;\yh1|\uh,\vh) - I(\yh2;\yh1|\uh,\vh) \notag\\
=& (\lambda+1) I(\yh1;\xh2|\uh)  + I(\yh1;\xh2|\uh,\vh)  .
\end{align}
Also, the RHS of \eqref{mainIneq} can be expressed as
\begin{align}
&(\lambda+1) I(\yh1;\phi_1|\uh)  +(\lambda+1) I(\yh2;\phi_2|\uh) - (\lambda+1) I(\yh1; \yh2|\uh) \notag\\
& +\lambda I(\xh1;\phi_1|\uh,\vh)  +\lambda I(\xh2;\phi_2|\uh,\vh) - \lambda I(\xh2;\xh1|\uh,\vh) \notag\\
= & (\lambda+1) I(\yh1;\xh2|\uh)  +(\lambda+1) I(\yh2;\yh1|\uh) - (\lambda+1) I(\yh1; \yh2|\uh) \\
& +\lambda I(\xh1;\xh2|\uh,\vh)  +\lambda I(\xh2;\yh1|\uh,\vh) - \lambda I(\xh2;\xh1|\uh,\vh) \notag\\
= & (\lambda+1) I(\yh1;\xh2|\uh)  +\lambda I(\yh1;\xh2|\uh,\vh). 
\end{align}
Substituting into \eqref{mainIneq}, we find that $(\lambda-1)  I(\yh1;\xh2|\uh,\vh) \leq 0 \Rightarrow  I(\yh1;\xh2|\uh,\vh)=0$.  Therefore, \eqref{mainIneq} is met with equality, and it  follows that:
\begin{align}
F(\lambda) & = I(\xh1;\uh|\xh2) - \lambda I(\yh1;\uh|\xh2) + I(\yh1;\vh|\uh,\xh2) - \lambda I(\xh1;\vh|\uh,\xh2) \\
&=  I(\xh2;\uh|\yh1) - \lambda I(\yh2;\uh|\yh1) + I(\yh2;\vh|\uh,\yh1) - \lambda I(\xh2;\vh|\uh,\yh1).
\end{align}
Since $\uh,\xh1,\yh1,\vh1$ conditioned on $\{\yh1=y\}$ is a candidate minimizer of \eqref{funL}, the second assertion of the claim follows.
By a symmetric argument, if we set $\phi_1 = \yh2, \phi_2 = \xh1$, the roles of the indices are reversed, and we find that 
\begin{align}
F(\lambda) & = I(\xh1;\uh|\yh2) - \lambda I(\yh1;\uh|\yh2) + I(\yh1;\vh|\uh,\yh2) - \lambda I(\xh1;\vh|\uh,\yh2) \\
&=  I(\xh2;\uh|\xh1) - \lambda I(\yh2;\uh|\xh1) + I(\yh2;\vh|\uh,\xh1) - \lambda I(\xh2;\vh|\uh,\xh1).
\end{align}
This establishes the first assertion of the claim and completes the proof.
\end{proof}

\begin{lemma} \label{lem:ConserveDerivative}
If $U_{\lambda},V_{\lambda}$ are valid minimizers of the functional \eqref{funL} for parameter $\lambda$, then
\begin{align}
I(Y;U_{\lambda}) + I(X;V_{\lambda}|U_{\lambda}) = -F'(\lambda)  ~~~\mbox{for a.e.\ $\lambda$.}
\end{align}
\end{lemma}
\begin{proof}
To begin, let  $U_{\lambda+\Delta},V_{\lambda+\Delta}$ be arbitrary, valid minimizers of the functional \eqref{funL} for parameter $\lambda+\Delta$, and let  $U_{\lambda-\Delta},V_{\lambda-\Delta}$ be arbitrary, valid minimizers of the functional  \eqref{funL} for parameter $\lambda-\Delta$.  Next, note that $F(\lambda)$ is concave and (strictly) monotone  decreasing in $\lambda$, and hence $F'(\lambda)$ exists for a.e. $\lambda$.  Thus, for any $\Delta > 0$, 
\begin{align}
\frac{F(\lambda+\Delta)-F(\lambda)}{\Delta} =& \frac{1}{\Delta} \Big(I(X;U_{\lambda+ \Delta}) + I(Y;V_{\lambda+ \Delta}|U_{\lambda+ \Delta}) - (\lambda+\Delta) \left( I(Y;U_{\lambda+ \Delta}) + I(X;V_{\lambda+ \Delta}|U_{\lambda+ \Delta})  \right) \Big)\notag\\
&  - \frac{1}{\Delta} \Big(I(X;U_{\lambda}) + I(Y;V_{\lambda}|U_{\lambda}) - \lambda \left( I(Y;U_{\lambda}) + I(X;V_{\lambda}|U_{\lambda})  \right) \Big) \notag\\
=& -\Big( I(Y;U_{\lambda+ \Delta}) + I(X;V_{\lambda+ \Delta}|U_{\lambda+ \Delta})  \Big) \\
& + \frac{1}{\Delta} \Big(I(X;U_{\lambda+ \Delta}) + I(Y;V_{\lambda+ \Delta}|U_{\lambda+ \Delta}) - \lambda \left( I(Y;U_{\lambda+ \Delta}) + I(X;V_{\lambda+ \Delta}|U_{\lambda+ \Delta})  \right) \Big)\notag\\
&  - \frac{1}{\Delta} \Big(I(X;U_{\lambda}) + I(Y;V_{\lambda}|U_{\lambda}) - \lambda \left( I(Y;U_{\lambda}) + I(X;V_{\lambda}|U_{\lambda})  \right) \Big) \notag\\
\geq& -\Big( I(Y;U_{\lambda+ \Delta}) + I(X;V_{\lambda+ \Delta}|U_{\lambda+ \Delta})  \Big),
\end{align}
 where the last inequality follows since $U_{\lambda + \Delta},V_{\lambda + \Delta}$ is a candidate minimizer of \eqref{funL} with parameter $\lambda$.
 
Similarly, 
\begin{align}
\frac{F(\lambda)-F(\lambda-\Delta)}{\Delta} =&  \frac{1}{\Delta} \Big(I(X;U_{\lambda}) + I(Y;V_{\lambda}|U_{\lambda}) - \lambda \left( I(Y;U_{\lambda}) + I(X;V_{\lambda}|U_{\lambda})  \right) \Big)\notag\\
&  - \frac{1}{\Delta} \Big(I(X;U_{\lambda-\Delta}) + I(Y;V_{\lambda-\Delta}|U_{\lambda-\Delta}) - (\lambda-\Delta) \left( I(Y;U_{\lambda-\Delta}) + I(X;V_{\lambda-\Delta}|U_{\lambda-\Delta})  \right) \Big)\notag\\
=& -\Big( I(Y;U_{\lambda- \Delta}) - I(X;V_{\lambda+ \Delta}|U_{\lambda+ \Delta})  \Big) \\
& + \frac{1}{\Delta} \Big(I(X;U_{\lambda}) + I(Y;V_{\lambda}|U_{\lambda}) - \lambda \left( I(Y;U_{\lambda}) + I(X;V_{\lambda}|U_{\lambda})  \right) \Big)\notag\\
&  - \frac{1}{\Delta} \Big(I(X;U_{\lambda-\Delta}) + I(Y;V_{\lambda-\Delta}|U_{\lambda-\Delta}) - \lambda \left( I(Y;U_{\lambda-\Delta}) + I(X;V_{\lambda-\Delta}|U_{\lambda-\Delta})  \right) \Big)\notag\\
\leq& -\Big( I(Y;U_{\lambda- \Delta}) + I(X;V_{\lambda- \Delta}|U_{\lambda- \Delta})  \Big),
\end{align}
 where the last inequality follows since $U_{\lambda - \Delta},V_{\lambda - \Delta}$ is a candidate minimizer of \eqref{funL} with parameter $\lambda$.  Recalling concavity of $F(\lambda)$, we have shown
\begin{align}
I(Y;U_{\lambda+ \Delta}) + I(X;V_{\lambda+ \Delta}|U_{\lambda+ \Delta})   \geq -F'(\lambda) \geq I(Y;U_{\lambda- \Delta}) + I(X;V_{\lambda- \Delta}|U_{\lambda- \Delta}).%
\end{align}
As $F'$ is monotone and well-defined up to a set of measure zero, we are justified in writing
\begin{align}
- \lim_{z\rightarrow \lambda^+}F'(z)\geq I(Y;U_{\lambda}) + I(X;V_{\lambda}|U_{\lambda}) \geq - \lim_{z\rightarrow \lambda^-}F'(z).
\end{align}
Since $F'$ is monotone, it is almost everywhere continuous, and so the LHS and RHS above coincide with $-F'(\lambda)$ for almost every $\lambda$.
 \end{proof}

 Since the derivative $F'(\lambda)$ is just a function of $F$ itself, and not of a particular minimizer, we have the following
\begin{corollary}\label{cor:Stability}
If $U_{\lambda},V_{\lambda}$ are valid minimizers of the functional \eqref{funL}  for parameter $\lambda$, then 
\begin{align}
I(X;{U}_{\lambda}|Y) + I(Y;{V}_{\lambda}|X) = F(\lambda)-(\lambda-1)F'(\lambda)  \mbox{~~for a.e. $\lambda$}.
\end{align}
\end{corollary}
\begin{proof}
Suppose $U_{\lambda},V_{\lambda}$ are valid minimizers.  Then, we can write:
\begin{align}
F(\lambda) &= I(X;U_{\lambda}) - \lambda I(Y;U_{\lambda}) + I(Y;V_{\lambda}|U_{\lambda})-\lambda I(X;V_{\lambda}|U_{\lambda})\\
&= I(X;U_{\lambda}) - \lambda I(Y;U_{\lambda}) + I(Y;V_{\lambda})-\lambda I(X;V_{\lambda}) + (\lambda-1)I(U_{\lambda};V_{\lambda})\\
&= I(X;U_{\lambda}|Y) + I(Y;V_{\lambda}|X)  +(\lambda-1)\Big(I(U_{\lambda};V_{\lambda}) - I(Y;U_{\lambda}) - I(X;V_{\lambda})\Big)\\
&=I(X;U_{\lambda}|Y) + I(Y;V_{\lambda}|X)  +(\lambda-1)F'(\lambda),
\end{align}
where the last line follows from Lemma \ref{lem:ConserveDerivative}.
\end{proof}

\begin{lemma} \label{lem:inductionLemma}
If $U,V$ are valid minimizers of the functional \eqref{funL}, and $\uh,\xh1,\yh1,\xh2,\yh2,\vh$ are constructed as described above, then there exist valid minimizers $\widetilde{U},\widetilde{V}$ such that
\begin{align}
I(X;U|Y) \geq I(X;\widetilde{U}|Y) + \frac{1}{2} I(\xh1;\xh2 | \uh, \yh1,\yh2).
\end{align}
\end{lemma}
\begin{proof}
To begin, note that:
\begin{align}
I(\xh1;\xh2|\uh,\yh1,\yh2) &= I(\xh1;\xh2,\uh|\yh1,\yh2) - I(\xh1;\uh|\yh1,\yh2)\\
&=I(\xh1;\uh|\yh1,\yh2,\xh2) - I(\xh1;\uh|\yh1,\yh2)\\
&=I(\xh1,\xh2;\uh|\yh1,\yh2) - I(\xh1;\uh|\yh1,\yh2) - I(\xh2;\uh|\yh1,\yh2) \\
&=2 I(X;U|Y) - I(\xh1;\uh|\yh1,\yh2) - I(\xh2;\uh|\yh1,\yh2). 
\end{align}
Thus, without loss of generality (relabeling indices 1 and 2 if necessary), we can assume 
\begin{align}
I(X;U|Y) \geq I(\xh1;\uh|\yh1,\yh2) + \frac{1}{2}I(\xh1;\xh2|\uh,\yh1,\yh2).%
\end{align}
Lemma \ref{lem:candidateMinimizers} asserts that, for almost every $y$, the tuple $\uh,\xh1,\yh1,\vh$ conditioned on $\{\yh2=y\}$ is a valid minimizer of \eqref{funL}.  Hence, there must exist a $y^*$ such that
\begin{align}
I(X;U|Y) \geq I(\xh1;\uh|\yh1,\yh2=y^*) + \frac{1}{2}I(\xh1;\xh2|\uh,\yh1,\yh2),
\end{align}
and  $\uh,\xh1,\yh1,\vh$ conditioned on $\{\yh2=y^*\}$ is a valid minimizer of \eqref{funL}.
Therefore, the claim follows by letting $\widetilde{U},X,Y,\widetilde{V}$ be equal in distribution to $\uh,\xh1,\yh1,\vh$ conditioned on $\{\yh2=y^*\}$.
\end{proof}

\begin{corollary} \label{cor:MIEqZero}
There exist $U,V$ which are valid minimizers of the functional \eqref{funL}, and satisfy 
\begin{align}
I(\xh1;\xh2 | \uh, \yh1,\yh2)=0,
\end{align}
where $\uh,\xh1,\yh1,\xh2,\yh2,\vh$ are constructed as described above.
\end{corollary}
\begin{proof}
Applying Lemma \ref{lem:inductionLemma}, we can inductively construct a sequence of valid minimizers $\{U^{(k)},X,Y,V^{(k)}\}_{k\geq 1}$ which satisfy 
\begin{align}
I(X;U^{(k)}|Y) \geq I(X;U^{(k+1)}|Y) + \frac{1}{2}I(\xh1;\xh2|\uh^{(k)},\yh1,\yh2) \mbox{~~for $k=1,2,\dots$,}
\end{align}
where $\uh^{(k)},\xh1,\yh1,\xh2,\yh2,\vh^{(k)} $ are constructed from two independent copies of $U^{(k)},X,Y,V^{(k)}$.
By Corollary \ref{cor:Stability}, we must also have
\begin{align}
I(X;U^{(k)}|Y) + I(Y;V^{(k)}|X) = F(\lambda)-(\lambda-1)F'(\lambda)
\end{align}
for all $k=1,2,\dots$. Therefore, for any $n$, we have:
\begin{align}
F(\lambda)-(\lambda-1)F'(\lambda) &= \frac{1}{n}\sum_{k=1}^n I(X;U^{(k)}|Y) + I(Y;V^{(k)}|X)\\
&\geq \frac{1}{n}\sum_{k=2}^n \Big( I(X;U^{(k)}|Y) + I(Y;V^{(k)}|X) \Big) + \frac{1}{n}\Big( I(X;U^{(n+1)}|Y) + I(Y;V^{(1)}|X)\Big) \\
&~~+ \frac{1}{2n}\sum_{k=1}^n I(\xh1;\xh2|\uh^{(k)},\yh1,\yh2)\notag\\
&\geq \frac{n-1}{n}\Big(F(\lambda)-(\lambda-1)F'(\lambda)  \Big) +  \frac{1}{2n}\sum_{k=1}^n I(\xh1;\xh2|\uh^{(k)},\yh1,\yh2),
\end{align}
and thus
\begin{align}
\sum_{k=1}^n I(\xh1;\xh2|\uh^{(k)},\yh1,\yh2) \leq 2\Big(F(\lambda)-(\lambda-1)F'(\lambda)  \Big). \label{sumConverges}
\end{align}
Hence, the sum on the LHS of \eqref{sumConverges} must converge as $n\rightarrow \infty$, implying
\begin{align}
\lim_{k\rightarrow \infty}   I(\xh1;\xh2|\uh^{(k)},\yh1,\yh2) = 0.
\end{align}
Arguing as in the proof of Lemma \ref{lem:InfAttained} in Appendix \ref{app:InfimaObtained}, we can conclude that there exists an optimizer $U,V$ for which $ I(\xh1;\xh2|\uh,\yh1,\yh2)$ is exactly zero.
\end{proof}

\begin{lemma}\cite{ghurye1962characterization}\label{SDcharacterization}
Let $\mathbf{A}_1$ and $\mathbf{A}_2$ be mutually independent $n$-dimensional random vectors.  If $\mathbf{A}_1+\mathbf{A}_2$ is independent of $\mathbf{A}_1-\mathbf{A}_2$, then  $\mathbf{A}_1$ and $\mathbf{A}_2$ are normally distributed.
\end{lemma}

\begin{corollary} \label{existUV_XuGauss}
There exist optimizers $U,V$ such that $X|\{U = u\}$ is Gaussian for a.e. $u$.
\end{corollary}
\begin{proof}
By construction and Corollary \ref{cor:MIEqZero}, we can conclude that there exist optimizers $U,V$ for which 
\begin{align}
I(X_1;X_2|U_1,U_2,Y_1,Y_2) = I(\xh1;\xh2|U_1,U_2,Y_1,Y_2) =0.
\end{align}
Therefore, by Lemma \ref{SDcharacterization}, there exist optimizers $U,V$ such that $X|\{U,Y = u,y\}$ is Gaussian for a.e. $u,y$.

Letting $P(x,y,u,v)$ denote the joint distribution of the above $X,Y,U,V$, we can use Markovity to write:
\begin{align}
P(x,y,u,v) &=P(u)P(y|u)P(x|u,y)P(v|y) \\
&= P(u)P(x|u)P(y|x)P(v|y).
\end{align}
Taking logarithms and rearranging, we have the identity
\begin{align}
\log(P(x|u)) = \log(P(y|u))+\log(P(x|u,y))-\log(P(y|x)).  \label{fromUVtoU}
\end{align}
Since $X|\{U,Y = u,y\}$ is Gaussian for a.e. $u,y$, and $X,Y$ are jointly Gaussian by assumption, the RHS of \eqref{fromUVtoU} is a quadratic function of $x$ for a.e. $u,y$.  Hence, $\log(P(x|u))$ is quadratic in $x$ for a.e. $u$, and the claim follows.  
\end{proof}

\begin{lemma}\label{OohamaLemma} \cite{Oohama1997}%
\label{lem:epi}
For any $U$ satisfying $U-X-Y$, the following inequality holds:
\begin{align}
 2^{-2 I(Y;U)}  \geq 1-\rho^2+\rho^2 2^{-2I(X;U)}.\label{oohamaEPI}
\end{align}
\end{lemma}
\begin{proof}
Consider any $U$ satisfying $U-X-Y$.  Let $Y_u, X_u$ denote the random variables $X,Y$ conditioned on $U=u$.  By Markovity and definition of $X,Y$, we have that $Y_u = \rho X_u + Z$, where $Z\sim N(0,1-\rho^2)$ is independent of $X_u$.  Hence, the conditional entropy power inequality implies that 
\begin{align*}
2^{2h(Y|U)} &\geq \rho^2 2^{2h(X|U)} + 2 \pi e(1-\rho^2) = 2 \pi e \rho^2 2^{-2I(X;U)} + 2 \pi e(1-\rho^2).
\end{align*}
From here, the lemma easily follows.
\end{proof}
\begin{lemma}\label{OohamaExplicit} 
\begin{align}
&\inf_{U: U-X-Y} \Big\{ I(X;U)-\lambda I(Y;U) \Big\} =  
\begin{cases}
\frac{1}{2}\left[\log\left(\frac{\rho^2(\lambda-1)}{1-\rho^2}\right)-\lambda\log 
\left(\frac{\lambda-1}{\lambda(1-\rho^2)}\right)\right] & \mbox{If $\lambda \geq 1/\rho^2$}\\
0 & \mbox{If $0 \leq \lambda \leq 1/\rho^2$.}
\end{cases}\label{uonly}
\end{align}
\end{lemma}
\begin{proof}
The claim follows from Lemma \ref{OohamaLemma} and Lemma \ref{calculusLemma} in Appendix \ref{app:Calculus} by identifying $a_1 \leftarrow \rho^2$ and $a_2 \leftarrow (1-\rho^2)$.
\end{proof}

\begin{lemma} \label{lem:characterizeF}
\begin{align}
F(\lambda) =\inf_{U: U-X-Y} \Big\{ I(X;U)-\lambda I(Y;U) \Big\}.
\end{align}
\end{lemma}
\begin{proof}
We will assume $\lambda \geq 1/\rho^2$.  The claim that $F(\lambda)=0$ in the complementary case will then follow immediately by monotonicity of $F(\lambda)$.  To this end, let $U,V$ be optimizers such that $X|\{U = u\}$ is Gaussian for a.e. $u$.  The existence of such $U,V$ is guaranteed by Corollary \ref{existUV_XuGauss}.  Let $X_u,Y_u$ denote the random variables $X,Y$ conditioned on $U=u$.  By Markovity, $X_u,Y_u$ are jointly Gaussian with
\begin{align}
Y_u = \rho X_u + Z,
\end{align}
where $Z\sim N(0,1-\rho^2)$ is independent of $X_u$.  Letting $\sigma_u^2$ be the variance of $X_u$, the variance of $Y_u$ is $\rho^2 \sigma_u^2 + (1-\rho^2)$.  Moreover, the squared linear correlation of $X_u$ and $Y_u$ is given by
\begin{align}
\rho_u^2 \triangleq \frac{\rho^2 \sigma_u^2}{\rho^2 \sigma_u^2 + (1-\rho^2)}.
\end{align}
 By Lemma  \ref{OohamaExplicit},
\begin{align}
&\inf_{V: V-Y_u-X_u} \Big\{ I(Y_u;V)-\lambda I(X_u;V) \Big\} =  \frac{1}{2}\left[\log\left(\frac{\rho_u^2(\lambda-1)}{1-\rho_u^2}\right)-\lambda\log
\left(\frac{\lambda-1}{\lambda(1-\rho_u^2)}\right)\right] \label{simpleOpt}
\end{align}
whenever $\lambda \geq 1/\rho_u^2$, and the infimum is equal to zero otherwise.

By definition, we have
\begin{align}
F(\lambda) &= I(X;U)- \lambda I(Y;U) + I(Y;V|U) - \lambda I(X;V|U) \\
&= \int  \Big(h(X)-h(X|u) - \lambda(h(Y) -h(Y|U=u))  + I(Y;V|U=u) - \lambda I(X;V|U=u) \Big)dP_U(u) \\
&= \int \Big(-\frac{1}{2}\log \sigma_u^2 + \frac{\lambda}{2}\log( \rho^2 \sigma_u^2 +(1-\rho^2))  + I(Y;V|U=u) - \lambda I(X;V|U=u) \Big)dP_U(u). \label{integrandScalar}
\end{align}
If $\lambda \geq 1/\rho_u^2$, we can apply \eqref{simpleOpt} to bound the integrand in \eqref{integrandScalar} as follows
\begin{align}
-\frac{1}{2}\log \sigma_u^2& + \frac{\lambda}{2}\log( \rho^2 \sigma_u^2 +(1-\rho^2))  + I(Y;V|U=u) - \lambda I(X;V|U=u)\notag\\
\geq & -\frac{1}{2}\log \sigma_u^2 + \frac{\lambda}{2}\log( \rho^2 \sigma_u^2 +(1-\rho^2)) + \frac{1}{2}\left[\log\left(\frac{\rho_u^2(\lambda-1)}{1-\rho_u^2}\right)-\lambda\log
\left(\frac{\lambda-1}{\lambda(1-\rho_u^2)}\right)\right] \\
=& \frac{1}{2}\left[\log\left(\frac{\rho^2(\lambda-1)}{1-\rho^2}\right)-\lambda\log
\left(\frac{\lambda-1}{\lambda(1-\rho^2)}\right)\right].
\end{align}
On the other hand, if $\lambda \leq 1/\rho_u^2$, then we can bound the integrand in \eqref{integrandScalar} by
\begin{align}
-\frac{1}{2}\log \sigma_u^2 & + \frac{\lambda}{2}\log( \rho^2 \sigma_u^2 +(1-\rho^2))  + I(Y;V|U=u) - \lambda I(X;V|U=u)\notag\\
\geq & -\frac{1}{2}\log \sigma_u^2 + \frac{\lambda}{2}\log( \rho^2 \sigma_u^2 +(1-\rho^2)) \\
\geq& \frac{1}{2}\left[\log\left(\frac{\rho^2(\lambda-1)}{1-\rho^2}\right)-\lambda\log
\left(\frac{\lambda-1}{\lambda(1-\rho^2)}\right)\right],
\end{align}
where the final inequality follows since $\lambda \leq 1/\rho_u^2\Rightarrow \sigma_u^2 \leq \frac{1-\rho^2}{\rho^2(\lambda-1)}$, and $-\frac{1}{2}\log \sigma_u^2 + \frac{\lambda}{2}\log( \rho^2 \sigma_u^2 +(1-\rho^2))$ is monotone decreasing in $\sigma_u^2$ for $ \sigma_u^2 \leq \frac{1-\rho^2}{\rho^2(\lambda-1)}$.  Therefore, we have established the inequality
\begin{align}
F(\lambda) \geq \frac{1}{2}\left[\log\left(\frac{\rho^2(\lambda-1)}{1-\rho^2}\right)-\lambda\log
\left(\frac{\lambda-1}{\lambda(1-\rho^2)}\right)\right]. %
\end{align}
The definition of $F(\lambda)$ together with Lemma \ref{OohamaExplicit} implies the reverse inequality,
completing the proof.
\end{proof}
Since \eqref{uonly} is a dual characterization of the inequality \eqref{oohamaEPI}, we have proved Theorem \ref{thm:scalar}.
\begin{remark}
 Although Lemma \ref{lem:characterizeF} implies that the functional \eqref{funL} is minimized when either $U$ or $V$ is degenerate, there are also minimizers for which this is not the case.  For example, if $-1\leq \rho_u,\rho_v \leq 1$ satisfy
\begin{align}
(1-\rho^2)(1-\rho^2\rho_u^2\rho_v^2) = \rho^2(\lambda-1)(1-\rho_u^2)(1-\rho_v^2),
\end{align}
then $U,V$ defined according to  
\begin{align}
U&= \rho_u X + Z_u\\
V&= \rho_v Y + Z_v,
\end{align}
where $Z_u \sim N(0,1-\rho_u^2)$ and  $Z_v \sim N(0,1-\rho_v^2)$ are  independent of everything else, also minimize \eqref{funL}. 
 \end{remark}

\section{Vector Setting}\label{sec:vector}
Now, we turn our attention to the vector case.  Throughout the remainder of this section, let $\Sigma_X, \Sigma_Z$ be positive definite $n\times n$ matrices.  Suppose $\mathbf{X}\sim N(\mu_X,\Sigma_X)$ and  $\mathbf{Z}\sim N(\mu_Z,\Sigma_Z)$ are independent $n$-dimensional Gaussian vectors, and define $\mathbf{Y} = \mathbf{X} + \mathbf{Z}$.  We recall the statement of Theorem \ref{thm:vector} here, along with conditions for equality:

\noindent \textbf{Theorem \ref{thm:vector}.}
\emph{
For any $U,V$ such that $U-\mathbf{X}-\mathbf{Y}-V$, 
\begin{align}
2^{-\frac{2}{n} (I(\mathbf{Y};U)+ I(\mathbf{X};V|U))} &\geq \frac{| \Sigma_X |^{1/n}}{| \Sigma_X  + \Sigma_Z |^{1/n}}  ~2^{-\frac{2}{n} (I(\mathbf{X};U)+I(\mathbf{Y};V|U))} + 2^{-\frac{2}{n} I(\mathbf{X};\mathbf{Y})}.\label{VecMain}%
\end{align}
Moreover, equality holds iff $\mathbf{X}|\{U=u\} \sim N(\mu_u ,\Sigma_{X|U} )$ for all $u$, where $\mu_u \triangleq \EE[\mathbf{X}|U=u]$, and $\Sigma_{X|U}$ is proportional to $\Sigma_Z$.}

\subsection{Proof of Theorem \ref{thm:vector}}
Instead of working directly with inequality \eqref{VecMain}, it will be convenient to consider a dual form.  As before, for $\lambda\geq 0$, define
\begin{align}
\mathbf{F}(\lambda)\triangleq \inf_{U,V: U-\mathbf{X}-\mathbf{Y}-V}\Big\{ I(\mathbf{X};U)-\lambda I(\mathbf{Y};U) + I(\mathbf{Y};V|U)-\lambda I(\mathbf{X};V|U)\Big\} \label{funLVector}
\end{align}
The remainder of this section is devoted to bounding the function $\mathbf{F}(\lambda)$.  

To begin, we remark that the extension of the results up to Lemma \ref{lem:epi} for the scalar setting immediately generalize to the present vector case by repeating the proofs verbatim.  Namely, we have the key observation:
\begin{corollary} \label{existUV_XuGaussVector}
There exist  $U,V$ which minimize the functional \eqref{funLVector} such that $\mathbf{X}|\{U = u\}$ is Gaussian for a.e. $u$.
\end{corollary}

Therefore, we pick up at this point and prove a vector version of Lemma \ref{lem:epi}.

\begin{lemma}\label{lem:vecEPI}
For any $U$ such that $U-\mathbf{X}-\mathbf{Y}$, 
\begin{align}
2^{-2 I(\mathbf{Y};U)/n} &\geq \frac{| \Sigma_X |^{1/n}}{| \Sigma_X  + \Sigma_Z |^{1/n}}  ~2^{-2 I(\mathbf{X};U)/n} + 2^{-2 I(\mathbf{X};\mathbf{Y})/n}.%
\end{align}
Moreover, equality holds iff $\mathbf{X}|\{U=u\} \sim N(\mu_u ,\Sigma )$ for all $u$, where $\mu_u \triangleq \EE[\mathbf{X}|U=u]$, and $\Sigma  \propto \Sigma_Z$.
\end{lemma}
\begin{proof}
Without loss of generality, we can assume $\mu_X=\mu_Z=0$.  Note that $\Sigma_Y = \Sigma_X  + \Sigma_Z$.  We can diagonalize the covariance matrices as 
\begin{align}
\Sigma_X &= U_X \Lambda_X U_X^T\\
\Sigma_Y &= U_Y \Lambda_Y U_Y^T,%
\end{align}
where $U_X$ and $U_Y$ are unitary.
Define $\mathbf{Y}' = \Lambda_Y^{-1/2} U_Y^T \mathbf{Y}$, implying $\mathbf{Y}'\sim N(0,I_n)$ and 
\begin{align}
\mathbf{Y}' = \Lambda_Y^{-1/2} U_Y^T  \mathbf{X} + \Lambda_Y^{-1/2} U_Y^T \mathbf{Z}.
\end{align}
Further, define $\mathbf{X}' = \Lambda_X^{-1/2} U_X^T \mathbf{X}$, so that $\mathbf{X}'\sim N(0,I_n)$ and 
\begin{align}
\mathbf{Y}' = \Lambda_Y^{-1/2} U_Y^T  U_X \Lambda_X^{1/2} \mathbf{X}' + \Lambda_Y^{-1/2} U_Y^T \mathbf{Z} = B \mathbf{X}' + W,
\end{align}
where $B \triangleq \Lambda_Y^{-1/2} U_Y^T  U_X \Lambda_X^{1/2}$ and $W\triangleq \Lambda_Y^{-1/2} U_Y^T \mathbf{Z}$, implying $W\sim N (0,\Lambda_Y^{-1/2} U_Y^T \Sigma_Z  U_Y \Lambda_Y^{-1/2} )$.
For any $U$ such that $U - \mathbf{X} - \mathbf{Y}$, $\mathbf{X}'$ and $W$ are independent given $U$.  Thus, the conditional entropy power inequality implies
\begin{align}
2^{2 h(\mathbf{Y}'|U)/n} &\geq 2^{2 h( B\mathbf{X}'|U)/n} + 2^{2 h( W | U)/n}\\
&= |B|^{2/n} 2^{2 h(\mathbf{X}'|U)/n} + 2 \pi e | \Lambda_Y^{-1/2} U_Y^T \Sigma_Z  U_Y \Lambda_Y^{-1/2}  |^{1/n} \\
&= | \Lambda_Y^{-1/2} U_Y^T  U_X \Lambda_X^{1/2} |^{2/n} 2^{2 h(\mathbf{X}'|U)/n} + 2 \pi e \frac{|\Sigma_Z|^{1/n}}{|\Sigma_Y|^{1/n}} \\
&=   \frac{|\Sigma_X|^{1/n}}{|\Sigma_Y|^{1/n}}  2^{2 h(\mathbf{X}'|U)/n} + 2 \pi e \frac{|\Sigma_Z|^{1/n}}{|\Sigma_Y|^{1/n}}.
\end{align}
Multiplying both sides by $2^{-2 h(\mathbf{Y}')/n} = 2^{-2 h(\mathbf{X}')/n} =\frac{1}{2\pi e}$, we obtain:
\begin{align}
2^{-2 I(\mathbf{Y}';U)/n} &\geq  \frac{|\Sigma_X|^{1/n}}{|\Sigma_Y|^{1/n}}  2^{-2 I(\mathbf{X}';U)/n} + \frac{|\Sigma_Z|^{1/n}}{|\Sigma_Y|^{1/n}}.
\end{align}
Since mutual information is invariant under one-to-one transformations of support, we also have
\begin{align}
2^{-2 I(\mathbf{Y};U)/n} &\geq    \frac{|\Sigma_X|^{1/n}}{|\Sigma_Y|^{1/n}}  2^{-2 I(\mathbf{X};U)/n} + \frac{|\Sigma_Z|^{1/n}}{|\Sigma_Y|^{1/n}}\\
&=\frac{| \Sigma_X |^{1/n}}{|\Sigma_X  + \Sigma_Z|^{1/n}}  2^{-2 I(\mathbf{X};U)/n} + 2^{-2 I(\mathbf{X};\mathbf{Y})/n}.%
\end{align}
The condition for equality follows from  the necessary conditions for equality in the conditional entropy power inequality.
\end{proof}

\begin{lemma} \label{lem:MatrixExplicit}Let $U$ be such that $U-\mathbf{X}-\mathbf{Y}$.
\begin{enumerate}
\item If $\lambda \geq  1 + |\Sigma_X^{-1} \Sigma_Z|^{1/n}$, then
\begin{align}
&I(\mathbf{X};U) - \lambda I(\mathbf{Y};U) \geq 
\frac{n}{2}\left[\log\left(\frac{|\Sigma_X |^{1/n}(\lambda-1)}{|\Sigma_Z|^{1/n}}\right)-\lambda\log
\left(\frac{| \Sigma_X + \Sigma_Z|^{1/n}(\lambda-1)}{ |\Sigma_Z|^{1/n} \lambda }\right)\right]. \label{matrixLBa}
\end{align}
\item If $0 \leq \lambda \leq  1 + |\Sigma_X^{-1} \Sigma_Z|^{1/n}$, then
\begin{align}
&I(\mathbf{X};U) - \lambda I(\mathbf{Y};U) \geq -\frac{\lambda n}{2}\log \left( \frac{|\Sigma_{X}+\Sigma_Z|^{1/n}}{|\Sigma_{X}|^{1/n} + |\Sigma_{Z}|^{1/n}}\right). \label{matrixLB2a}
\end{align}
\end{enumerate}
\end{lemma}
\begin{proof}
The claim follows from Lemma \ref{calculusLemma} in Appendix \ref{app:Calculus} by identifying $a_1 \leftarrow \frac{|\Sigma_X|^{1/n}}{|\Sigma_X+\Sigma_Z|^{1/n}}$ and $a_2 \leftarrow 2^{-2 I(\mathbf{X};\mathbf{Y})/n} = \frac{|\Sigma_Z|^{1/n}}{|\Sigma_X+\Sigma_Z|^{1/n}}$.  The hypothesis that $a_1 + a_2 \leq 1$ is satisfied since the Minkowski determinant theorem \cite{marcus1992survey} asserts that $|\Sigma_X|^{1/n} + |\Sigma_Z|^{1/n} \leq |\Sigma_X+\Sigma_Z|^{1/n}$.
\end{proof}

Note that the lower bound \eqref{matrixLBa} is achieved if $U$ can be chosen such that $\mathbf{X}|\{U=u\} \sim N (\mu_u , \Sigma_{X|U})$ for each $u$, and $ \Sigma_{X|U}  = \alpha \Sigma_{Z}$.  In this case
\begin{align}
\frac{1}{n}I(\mathbf{X};U) - \frac{\lambda}{n} I(\mathbf{Y};U) &= \frac{1}{2n}\log\left( \frac{|\Sigma_{X}|}{|\Sigma_{X|U}|}\right) - \frac{\lambda}{2n}\log \left( \frac{|\Sigma_{X}+\Sigma_Z|}{|\Sigma_{X|U} + \Sigma_Z |} \right)   \\
&= \frac{1}{2n}\log\left( 
\frac{|\Sigma_{X}|}
{\alpha^n |\Sigma_Z|}
\right) 
- \frac{\lambda}{2n}\log \left( 
\frac{|\Sigma_{X}+\Sigma_Z|}
{(1+\alpha)^n |\Sigma_Z |}
 \right) \\
 &= \frac{1}{2}\left[\log\left(\frac{| \Sigma_X |^{1/n}(\lambda-1)}{|\Sigma_Z|^{1/n}}\right)-\lambda\log
\left(\frac{|\Sigma_X + \Sigma_Z|^{1/n}(\lambda-1)}{ |\Sigma_Z|^{1/n} \lambda }\right)\right],
\end{align}
where we set $\alpha = \frac{1}{\lambda-1}$ to arrive at the final equality.  The lower bound \eqref{matrixLB2a} is only attainable if $\Sigma_X$ and $\Sigma_Z$ are proportional.  In this case, the RHS of \eqref{matrixLB2a} is precisely zero.

\begin{lemma} \label{lem:FVecExplicit}
~

\begin{enumerate}
\item If $\lambda \geq  1 + |\Sigma_X^{-1} \Sigma_Z|^{1/n}$, then
\begin{align}
&\mathbf{F}(\lambda) \geq 
\frac{n}{2}\left[\log\left(\frac{|\Sigma_X |^{1/n}(\lambda-1)}{|\Sigma_Z|^{1/n}}\right)-\lambda\log
\left(\frac{| \Sigma_X + \Sigma_Z|^{1/n}(\lambda-1)}{ |\Sigma_Z|^{1/n} \lambda }\right)\right]. \label{matrixLB}
\end{align}
\item If $0 \leq \lambda \leq  1 + |\Sigma_X^{-1} \Sigma_Z|^{1/n}$, then
\begin{align}
&\mathbf{F}(\lambda) \geq -\frac{\lambda n}{2}\log \left( \frac{|\Sigma_{X}+\Sigma_Z|^{1/n}}{|\Sigma_{X}|^{1/n} + |\Sigma_{Z}|^{1/n}}\right). \label{matrixLB2}
\end{align}
\end{enumerate}
\end{lemma}
\begin{proof}
Similar to the scalar setting, it is sufficient to restrict our attention to the setting where $\lambda \geq  1 + |\Sigma_X^{-1} \Sigma_Z|^{1/n}$.  Therefore, we assume this throughout the proof.

Let $U,V$ be valid minimizers of \eqref{funLVector}, where $\mathbf{X}|\{U=u\}$ is Gaussian for a.e. $u$.  The existence of such $U,V$ is guaranteed by Corollary \ref{existUV_XuGaussVector}. By definition, we have
\begin{align}
\mathbf{F}(\lambda) &= I(\mathbf{X};U)- \lambda I(\mathbf{Y};U) + I(\mathbf{Y};V|U) - \lambda I(\mathbf{X};V|U) \\
&= \int  \Big(h(\mathbf{X})-h(\mathbf{X}|u) - \lambda(h(\mathbf{Y}) -h(\mathbf{Y}|U=u))  + I(\mathbf{Y};V|U=u) - \lambda I(\mathbf{X};V|U=u) \Big)dP_U(u) \label{MatrixIntegrand}
\end{align}
  
Let $\mathbf{X}_u$ and $\mathbf{Y}_u$ denote the random variables $\mathbf{X},\mathbf{Y}$ conditioned on $\{U=u\}$.  Suppose $\mathbf{X}_u \sim N(\mu_u, \Sigma_{X_u})$.  Then $\mathbf{X}_u,\mathbf{Y}_u$ are jointly normal, and we can write
$ \mathbf{X}_u =  B \mathbf{Y}_u + W$, where $W\sim N(0, \Sigma_{X_u }-\Sigma_{X_u Y_u}  \Sigma^{-1}_{ Y_u} \Sigma_{ Y_u X_u} )$ is independent of $\mathbf{Y}_u$, and $B = \Sigma_{X_u Y_u}  \Sigma^{-1}_{ Y_u}$.  Note that Markovity implies $\mathbf{Y}_u = \mathbf{X}_u + \mathbf{Z}$, so that $\Sigma_{X_u Y_u} = \Sigma_{Y_u X_u} = \Sigma_{X_u}$, which simplifies $\Sigma_{X_u Y_u}  \Sigma^{-1}_{ Y_u} \Sigma_{ Y_u X_u}$ to $\Sigma_{X_u}  \Sigma^{-1}_{ Y_u} \Sigma_{X_u}$, and also implies $\Sigma_{Y_u}-\Sigma_{X_u}=\Sigma_Z$.

Suppose $\Sigma_{X_u}$ is such that $\lambda \geq  1 + \frac{|\Sigma_{X_u }-\Sigma_{X_u }  \Sigma^{-1}_{ Y_u} \Sigma_{  X_u}|^{1/n}}{|\Sigma_{X_u }  \Sigma^{-1}_{ Y_u} \Sigma_{  X_u}|^{1/n}} = 1 + \frac{|\Sigma_Z|^{1/n}}{|\Sigma_{X_u}|^{1/n}}$.  Then, Lemma \ref{lem:MatrixExplicit} allows us to 
lower bound the integrand in \eqref{MatrixIntegrand} as:
\begin{align}
h(\mathbf{X}) - h(\mathbf{X}|u)& - \lambda(h(\mathbf{Y}) - h(\mathbf{Y}|u)) +  I(\mathbf{Y};V|U=u) - \lambda I(\mathbf{X};V|U=u) \notag\\
\geq& 
\frac{n}{2}\left[ \log\left( \frac{|\Sigma_X|^{1/n}}{|\Sigma_{X_u}|^{1/n}}\right) - \lambda \log\left( \frac{|\Sigma_Y|^{1/n}}{|\Sigma_{Y_u}|^{1/n}}\right)\right]\\
&+\frac{n}{2}\left[\log\left(\frac{|B \Sigma_{Y_u} B^T |^{1/n}(\lambda-1)}{|\Sigma_W |^{1/n}}\right)-\lambda\log
\left(\frac{|B \Sigma_{Y_u} B^T+  \Sigma_W |^{1/n}(\lambda-1)}{ | \Sigma_W |^{1/n} \lambda }\right)\right]\notag\\
=&\frac{n}{2}\left[ \log\left( \frac{|\Sigma_X|^{1/n}}{|\Sigma_{X_u}|^{1/n}}\right) - \lambda \log\left( \frac{|\Sigma_Y|^{1/n}}{|\Sigma_{Y_u}|^{1/n}}\right)\right]\\
&+\frac{n}{2}\left[\log\left(\frac{|\Sigma_{X_u }  \Sigma^{-1}_{ Y_u} \Sigma_{  X_u}|^{1/n}(\lambda-1)}{|\Sigma_{X_u }-\Sigma_{X_u }  \Sigma^{-1}_{ Y_u} \Sigma_{  X_u}|^{1/n}}\right)
-\lambda\log
\left(\frac{|\Sigma_{Xu}|^{1/n}(\lambda-1)}{ |\Sigma_{X_u }-\Sigma_{X_u }  \Sigma^{-1}_{ Y_u} \Sigma_{  X_u} |^{1/n} \lambda }\right)\right]\notag\\
=&\frac{n}{2}\left[ \log\left( \frac{|\Sigma_X|^{1/n}}{|\Sigma_{X_u}|^{1/n}}\right) - \lambda \log\left( \frac{|\Sigma_Y|^{1/n}}{|\Sigma_{Y_u}|^{1/n}}\right)\right]\\
&+\frac{n}{2}\left[\log\left(\frac{|\Sigma_{X_u } |^{1/n}(\lambda-1)}{|\Sigma_{Y_u }-\Sigma_{X_u } |^{1/n}}\right)
-\lambda\log
\left(\frac{|\Sigma_{Y_u}|^{1/n}(\lambda-1)}{ |\Sigma_{Y_u }-\Sigma_{X_u } |^{1/n} \lambda }\right)\right]\notag\\
=&\frac{n}{2}\left[
\log\left(
\frac{|\Sigma_X|^{1/n} (\lambda-1)}
{  |\Sigma_{Z }|^{1/n}}\right)
-\lambda\log\left(
\frac{|\Sigma_{X}+\Sigma_Z|^{1/n}(\lambda-1)}
{ |\Sigma_{Z} |^{1/n} \lambda }
\right)\right]. \label{finalDesired}
\end{align}
On the other hand, if $\Sigma_{X_u}$ is such that $0 \leq \lambda \leq 1 + \frac{|\Sigma_Z|^{1/n}}{|\Sigma_{X_u}|^{1/n}}$.  Then, Lemma \ref{lem:MatrixExplicit} allows us to 
lower bound the integrand in \eqref{MatrixIntegrand} as:
\begin{align}
h(\mathbf{X}) - h(\mathbf{X}|u)& - \lambda(h(\mathbf{Y}) - h(\mathbf{Y}|u)) +  I(\mathbf{Y};V|U=u) - \lambda I(\mathbf{X};V|U=u) \notag \\
\geq& 
\frac{n}{2}\left[ \log\left( \frac{|\Sigma_X|^{1/n}}{|\Sigma_{X_u}|^{1/n}}\right) - \lambda \log\left( \frac{|\Sigma_Y|^{1/n}}{|\Sigma_{Y_u}|^{1/n}}\right)\right]\\
&-\frac{\lambda n}{2}\log \left( \frac{|\Sigma_{X_u}|^{1/n}}{|\Sigma_{X_u}\Sigma_{Y_u}^{-1}\Sigma_{X_u}|^{1/n} + |\Sigma_{X_u} - \Sigma_{X_u}\Sigma_{Y_u}^{-1}\Sigma_{X_u}|^{1/n}}\right)\notag\\
=& 
\frac{n}{2}\left[ \log\left( \frac{|\Sigma_X|^{1/n}}{|\Sigma_{X_u}|^{1/n}}\right) 
- \lambda \log \left( \frac{|\Sigma_Y|^{1/n}}{|\Sigma_{X_u} |^{1/n} + |\Sigma_{Z}|^{1/n}}\right)\right]\\
\geq & 
\frac{n}{2}\left[ \log\left( \frac{|\Sigma_X|^{1/n}(\lambda-1) }{|\Sigma_{Z}|^{1/n}}\right) 
- \lambda \log \left( \frac{|\Sigma_Y|^{1/n}(\lambda-1)}{ |\Sigma_{Z}|^{1/n} \lambda }\right)\right],
\end{align}
where the last line follows since $ -\log\left({|\Sigma_{X_u}|^{1/n}}\right) +
 \lambda \log \left( {|\Sigma_{X_u} |^{1/n} + |\Sigma_{Z}|^{1/n}}\right)$ is a monotone decreasing function of $|\Sigma_{X_u} |^{1/n}$ provided  $\lambda \leq 1 + \frac{|\Sigma_Z|^{1/n}}{|\Sigma_{X_u}|^{1/n}}$.  Thus, setting $|\Sigma_{X_u} |^{1/n} = \frac{1}{\lambda-1}|\Sigma_{Z} |^{1/n}$ only weakens the inequality.
 \end{proof}
The combination of Lemmas \ref{lem:vecEPI}-\ref{lem:FVecExplicit} proves Theorem \ref{thm:vector}.
\section{Closing Remarks}\label{sec:conclusion}
The focus of this paper was on  the extremal result asserted by Theorem \ref{thm:vector}, and not on operational coding problems.  However, since the entropy-power-like inequality of Theorem \ref{thm:vector}  leads to what is arguably the simplest solution for the two-encoder quadratic Gaussian source coding problem (an archetypical problem in network information theory), we have little doubt that it will find other interesting applications.  We provided Theorem \ref{thm:Ellipsoid} as one such example.  As another example, Theorem \ref{thm:scalar} can be applied to show that jointly Gaussian auxiliaries exhaust the rate region for multiterminal source coding under logarithmic loss \cite{bib:CourtadeWeissman_LogLoss_TransIT2014} when the sources are Gaussian.  This leads to yet another  solution for the  two-encoder quadratic Gaussian source coding problem, and unifies the two problems under the paradigm of compression under logarithmic loss.

\section*{Acknowledgement}
The authors acknowledge several conversations with Tsachy Weissman, Kartik Venkat, and Vignesh  Ganapathi-Subramanian which led to deeper intuition and insight. %
The first author also wishes to acknowledge an inspiring  discussion with Chandra Nair that took place at the 2014 International Zurich Seminar on Communications following the presentation of \cite{bib:CourtadeIZS}.  This conversation and the following exchanges generated the spark which led to a successful proof of  Theorem \ref{thm:scalar}, a long-held conjecture of the first author that was first made public  on the \emph{Information Theory b-Log} in March, 2013 \cite{BLOG, ITSOCnewsletter}.

\appendices

\section{Proof of Theorem \ref{thm:Ellipsoid}}\label{app:ellipsoidProof}

\noindent \textbf{Converse Part:} 

Define $U=f_x(\vecX_1, \dots, \vecX_k)$, and $V=f_y(\vecY_1, \dots, \vecY_k)$, and suppose that $U,V$ are such that we can determine an ellipsoid $\mathcal{E}(A_x,b_x)$ having volume bounded by
\begin{align}
\Big( \operatorname{vol}(\mathcal{E}(A_x,b_x) ) \Big)^{1/n}= \frac{c_n^{1/n}}{|A_x|^{1/n}} \leq {c_n^{1/n}}{\sqrt{{n \nu_x}  |\Sigma_n|^{1/n}}} \label{VolSupp}
\end{align}
and containing the points $\{\mathbf{X}_i\}_{i=1}^k$ with probability at least $1-\epsilon_n$, where $\lim_{n\rightarrow\infty} \epsilon_n= 0$.

For $i\in \{1,\dots,k\}$, define the indicator random variable $E_i = \mathbf{1}_{\{\mathbf{X}_i\notin \mathcal{E}(A_x,b_x) \}}$, and note that 
\begin{align}
&\left| \mathbb{E}\left[ \left(A_x \mathbf{X}_i - \mathbb{E}[A_x\mathbf{X}_i|U,V,E_i] \right)\left(A_x \mathbf{X}_i - \mathbb{E}[A_x\mathbf{X}_i|U,V,E_i] \right)^T \Big| E_i=0\right]\right|^{1/n} \notag\\
&\leq \frac{1}{n}{\Tr\left( \mathbb{E}\left[ \left(A_x \mathbf{X}_i - \mathbb{E}[A_x\mathbf{X}_i|U,V,E_i] \right)\left(A_x \mathbf{X}_i - \mathbb{E}[A_x\mathbf{X}_i|U,V,E_i] \right)^T \Big| E_i=0 \right] \right)}\label{AGineq}\\
&= \frac{1}{n}\mathbb{E}\left[ \left\| A_x \mathbf{X}_i - \mathbb{E}[A_x\mathbf{X}_i|U,V,E_i] \right\|^2   \Big| E_i=0 \right]\\
& \leq \frac{1}{n}\mathbb{E}\left[ \left\| A_x \mathbf{X}_i -b_x\right\|^2  \Big| E_i=0 \right]\label{mmseEst}\\
&\leq \frac{1}{n},\label{EllipsHyp}
\end{align}
where \eqref{AGineq} follows from the inequality of arithmetic and geometric means, \eqref{mmseEst} follows since conditional expectation minimizes mean square error, and \eqref{EllipsHyp} follows since $\{E_i=0\}\Rightarrow \{\mathbf{X}_i \in \mathcal{E}(A_x,b_x)\}\Rightarrow  \{\left\| A_x \mathbf{X}_i -b_x\right\|\leq 1\}$.  Summarizing the above, we have established that, for each $i=1,\dots,k$, 
\begin{align}
\left| \mathbb{E}\left[ \left(A_x \mathbf{X}_i - \mathbb{E}[A_x\mathbf{X}_i|U,V,E_i] \right)\left(A_x \mathbf{X}_i - \mathbb{E}[A_x\mathbf{X}_i|U,V,E_i] \right)^T \Big| E_i=0\right]\right|^{1/n} \leq \frac{1}{n}. \label{maxEntLem}
\end{align}
Applying \eqref{maxEntLem} in conjunction with the maximum-entropy property of Gaussians, we have
\begin{align}
\frac{1}{n}I(\mathbf{X}_i;U,V|E_i=0) &= \frac{1}{n}h(\mathbf{X}_i|E_i=0)- \frac{1}{n}h(\mathbf{X}_i|U,V,E_i=0)\\
&=\frac{1}{n}h(\mathbf{X}_i|E_i=0)- \frac{1}{n}h\left(A_x \mathbf{X}_i|U,V,E_i=0\right) + \frac{1}{n}\log |A_x | \\
&\geq \frac{1}{n}h(\mathbf{X}_i|E_i=0)  - \frac{1}{2n}\log\left( (2\pi e)^n\right) - \frac{1}{2}\log\left( \frac{1}{n}\right) + \frac{1}{2n}\log \frac{1}{n^n \nu_x^n |\Sigma_n|} \label{keyVolIneq}\\
&= \frac{1}{n}\Big( h(\mathbf{X}_i|E_i=0)-h(\vecX_i)\Big) + \frac{1}{2}\log \frac{1}{\nu_x}.
\end{align}
Since $I(\vecX_i;U,V)\geq I(\vecX_i;U,V|E_i)-H(E_i)\geq \Pr\{E_i=0\}I(\vecX_i;U,V|E_i=0) - H(E_i)$ and $\Pr\{E_i=0\} \geq 1-\epsilon_n \rightarrow 1$, it follows by Lemma \ref{lem:hXE_hX} (see Appendix \ref{app:Calculus}) that, for any $\delta>0$,
\begin{align}
\frac{1}{n} I(\vecX_i;U,V) \geq \frac{1}{2}\log \frac{1}{\nu_x\sqrt{1+\delta}}
\end{align} 
for $n$ sufficiently large.  Since $\vecX_1, \dots, \vecX_k$ are mutually independent, we have
\begin{align}
\frac{1}{k n} I(\vecX_1, \dots, \vecX_k;U,V) \geq \frac{1}{k n} \sum_{i=1}^k I(\vecX_i;U,V) \geq \frac{1}{2}\log \frac{1}{\nu_x\sqrt{1+\delta}}\label{xCons}
\end{align}
for all $n$ sufficiently large.
By a symmetric argument, it also holds that
\begin{align}
\frac{1}{k n} I(\vecY_1, \dots, \vecY_k;U,V) \geq \frac{1}{2}\log \frac{1}{\nu_y\sqrt{1+\delta}}\label{yCons}
\end{align}
for  $n$ sufficiently large.  Thus, since 
\begin{align}
R_x+R_y \geq \frac{1}{k n} H(U,V) \geq \frac{1}{k n} I(\vecX_1,\vecY_1, \dots, \vecX_k,\vecY_k;U,V),
\end{align}
it follows from \eqref{eqnRD} and \eqref{RDineq} that 
\begin{align}
R_x + R_y \geq \frac{1}{2}\log \frac{(1-\rho^2)\beta((1+\delta)\nu_x \nu_y)}{2(1+\delta)\nu_x \nu_y}.
\end{align}
Since $\delta>0$ can be taken arbitrarily small, \eqref{RxyConstratint} must hold.
Likewise, since \eqref{introWhite} implies
\begin{align}
2^{\frac{2}{kn} \left(I(\mathbf{X}_1, \dots, \vecX_k;U|V)-I(\mathbf{X}_1, \dots, \vecX_k;U,V)\right)}  = 2^{-\frac{2}{kn} I(\mathbf{X}_1, \dots, \vecX_k;V)}  &\geq (1-\rho^2) + \rho^2 2^{-\frac{2}{kn} I(\mathbf{Y}_1, \dots, \vecY_k;V)},%
\end{align}
it follows that
\begin{align}
2^{\frac{2}{kn} \left(H(U)-I(\mathbf{X}_1, \dots, \vecX_k;U,V)\right)}  &\geq (1-\rho^2) + \rho^2 2^{-\frac{2}{kn} H(V)}.%
\end{align}
Therefore,   \eqref{xCons} also implies \eqref{RxConstratint}.  By a symmetric argument, we obtain \eqref{RyConstratint}.

\medskip
\noindent \textbf{Direct Part:}

Fix $\epsilon>0$, and assume \eqref{RxConstratint}-\eqref{RxyConstratint} are satisfied.  Further, diagonalize $\Sigma = U_{\Sigma} \Lambda U_{\Sigma}^T$ (throughout the proof, we suppress the explicit dependence of the covariance matrices on $n$ for convenience).

Suppose $(X^n,Y^n)$ is a sequence of independent pairs of random variables, where $(X_j,Y_j)$ are jointly normal with linear correlation $\rho$, and $X_i,Y_i$ each have unit variance.
For $n$ sufficiently large, the achievability result for the two-encoder quadratic Gaussian source coding problem (e.g., \cite[Theorem 1]{WagnerRateRegion2008}) implies that there exist functions $f_x : \mathbb{R}^{n} \rightarrow \{1,2,\dots,2^{n R_x }\}$ and $f_y : \mathbb{R}^{n} \rightarrow \{1,2,\dots,2^{n R_y }\}$
for which
\begin{align}
\Pr\left\{ \mathbb{E} \left[ \| X^n -   \mathbb{E}\left[ X^n |f_x(X^n),f_y(Y^n)\right]  \|^2\right] > n\nu_x  \right\} &< \epsilon\quad\mbox{and}  \\
\Pr\left\{ \mathbb{E} \left[ \| Y^n -   \mathbb{E}\left[  Y^n |f_x(X^n),f_y(Y^n)\right]  \|^2\right] > n\nu_y  \right\} &< \epsilon.
\end{align}
Therefore, the rates $(R_x,R_y)$ are sufficient to communicate $k$ ellipsoids of the form
\begin{align}
\mathcal{E}_i = \{ x : \|A_x x - b_i\| \leq 1 \}  \mbox{~~~~for $i=1,2,\dots, k$,}
\end{align}
where $A_x \triangleq \frac{1}{\sqrt{n \nu_x } }\Lambda^{-1/2}U_{\Sigma}^T$,  $b_i \triangleq   A_x \mathbb{E}\left[  {\mathbf{X}_i}|f_x({\mathbf{X}_i}),f_y({\mathbf{Y}_i})\right]$, and $\vecX_i \in \mathcal{E}_i$ with probability greater than $1-\epsilon$ for $i=1,2,\dots, k$.  This follows since the pair of vectors $(\Lambda^{-1/2}U_{\Sigma}^T \vecX_i, \Lambda^{-1/2}U_{\Sigma}^T \vecY_i)$ is equal in distribution to $(X^n,Y^n)$.

Let $\delta>0$ satisfy $\sqrt{\delta/\nu_x} < \epsilon$, and define $\tau = 1-2\sqrt{\delta}$, and $\gamma = (1-\delta)\tau$ for convenience.  Further, let $u_1, \dots, u_k$ be an orthonormal basis for the vector space spanned by $ b_1, \dots,  b_k$, and define 
\begin{align}
\matr{\mathbf{B_x}} \triangleq
\left(\tau I_n - \gamma \sum_{i=1}^k u_i u_i^T\right)A_x.
\end{align}
We remark that $\matr{\mathbf{B_x}}$ is a random matrix, and is a function of $\{f_x({\mathbf{X}_i}),f_y({\mathbf{Y}_i})\}_{i=1}^k$.  Note that 
\begin{align}
\matr{\mathbf{B_x}} \vecX_i &= \left(\tau I_n - \gamma \sum_{i=1}^k u_i u_i^T\right) A_x \vecX_i = \left(\tau I_n - \gamma \sum_{i=1}^k u_i u_i^T\right) \widetilde{\vecX}_i , 
\end{align}
where $\widetilde{\vecX}_i \triangleq A_x \vecX_i$ for $i=1,2, \dots, k$.  Define $\mathbf{Z}_i=\widetilde{\vecX}_i - b_i$, and continue with
\begin{align}
\matr{\mathbf{B_x}} \vecX_i 
= \left(\tau I_n - \gamma \sum_{j=1}^k u_j u_j^T\right) \widetilde{\vecX}_i 
&=\tau ( \widetilde{\vecX}_i - b_i ) + \tau b_i - \gamma \sum_{j=1}^k u_j u_j^T\left(\mathbf{Z}_i +b_i\right) \\
&=\tau ( \widetilde{\vecX}_i - b_i ) + (\tau-\gamma) b_i - \gamma \sum_{j=1}^k u_j u_j^T \mathbf{Z}_i,\label{orthProj}
\end{align}
where \eqref{orthProj} follows since $\sum_{j=1}^k u_j u_j^T$ is an orthogonal projection onto the space spanned by $b_1, \dots, b_k$.
Applying the triangle inequality, we can conclude
\begin{align}
\|\matr{\mathbf{B_x}} \vecX_i \| &\leq \tau \| \widetilde{\vecX}_i - b_i \| + (\tau-\gamma) \|b_i \| + \gamma \left\| \sum_{j=1}^k u_j u_j^T \mathbf{Z}_i \right\|\\
&= \tau \| \widetilde{\vecX}_i - b_i \| + (\tau-\gamma) \|b_i \| + \gamma \sqrt{\sum_{j=1}^k (u_j^T \mathbf{Z}_i)^2}\\
&\leq \tau \| \widetilde{\vecX}_i - b_i \| + \delta  \|b_i \| +  \sqrt{\sum_{j=1}^k (u_j^T \mathbf{Z}_i)^2}.
\end{align}
Using Jensen and Markov's inequalities, we can bound
\begin{align}
\Pr\left\{ \|b_i \|  \geq \frac{1}{\sqrt{\delta}} \right\} &= \Pr\left\{ \left\|\mathbb{E}\left[  {\widetilde{\mathbf{X}}_i}|f_x({\mathbf{X}_i}),f_y({\mathbf{Y}_i})\right] \right\|  \geq \frac{1}{\sqrt{\delta}} \right\}\\
&\leq  \Pr\left\{ \mathbb{E}\left[  \left\|{\widetilde{\mathbf{X}}_i} \right\| \Big| f_x({\mathbf{X}_i}),f_y({\mathbf{Y}_i})\right]  \geq \frac{1}{\sqrt{\delta}} \right\}\\
&\leq  \sqrt{\delta} \mathbb{E}\left[  \left\|{\widetilde{\mathbf{X}}_i} \right\|\right] \\
&\leq \sqrt{\delta\mathbb{E}\left[  \left\|{\widetilde{\mathbf{X}}_i} \right\|^2 \right] }\\
&= \sqrt{\frac{\delta}{\nu_x}}\label{usedXDist}\\
&\leq \epsilon,
\end{align}
where \eqref{usedXDist} follows since  $\widetilde{X}_i \sim N\left(0,\frac{1}{n\nu_x} I_n\right)$.

Next, since $b_j$ is the LLMSE estimator of $\widetilde{X}_i$ given $f_x({\mathbf{X}_i}),f_y({\mathbf{Y}_i})$ by construction, we have
\begin{align}
\mathbb{E} \left[  \sum_{j=1}^k (u_j^T \mathbf{Z}_i)^2 \right] &=\sum_{j=1}^k u_j^T\Sigma_{\mathbf{Z}_i} u_j \leq \frac{k}{n \nu_x},
\end{align}
where the  inequality follows since the center quantity is upper bounded by the $k$ largest eigenvalues of $\Sigma_{\mathbf{Z}_i}$, which are themselves upper bounded by the $k$ largest eigenvalues of $\Sigma_{\widetilde{\vecX}_i} = \frac{1}{n\nu_x}I_n$.  Therefore, proceeding with Markov's inequality, we have
\begin{align}
\Pr\left\{ \sum_{j=1}^k (u_j^T \mathbf{Z}_i)^2 \geq \delta \right\} \leq \frac{k }{\delta n \nu_x},
\end{align}
 which is upper-bounded by $\epsilon$ for $n$ sufficiently large.

Also by construction, $\| \widetilde{\vecX}_i - b_i \|  =  \| A_x {\vecX}_i - b_i \| \leq 1$ with probability greater than $1-\epsilon$.  Therefore for $n$ sufficiently large, we can conclude that
\begin{align}
\Pr \Big\{ \|\matr{\mathbf{B_x}} \vecX_i \| \leq 1 \Big\}  > 1-3\epsilon .%
\end{align} 
Finally, note that 
\begin{align}
\left| \matr{\mathbf{B_x}}\right| &= \left| A_x \right| \left| \tau I_n - \gamma \sum_{i=1}^k u_i u_i^T \right| =  \left| A_x\right| \tau^{n-k}(\tau-\gamma)^k,
\end{align}
and so
\begin{align}
\left| \matr{\mathbf{B_x}}\right| ^{1/n} = \tau \delta^{k/n} |A_x|^{1/n} = \frac{(1-2\sqrt{\delta}) \delta^{k/n}}{\sqrt{n \nu_x|\Sigma|^{1/n}}}.
\end{align}
Since $\delta$ can be taken arbitrarily small, a symmetric argument involving the $\vecY_i$'s completes the proof.

\section{Existence of Minimizers}\label{app:InfimaObtained}
We will require the following result on lower semicontinuity of relative entropy.
\begin{lemma} \cite[Lemma 1.4.3]{dupuis2011weak}\label{lem:DSemiCont}
Let $\mathcal{X}$ be a Polish space, and let $\mathcal{P}(\mathcal{X})$ denote the set of probability measures on $\mathcal{X}$.  The relative entropy $D(P \|Q)$ is a lower semicontinuous function of $(P,Q) \in \mathcal{P}(\mathcal{X})\times \mathcal{P}(\mathcal{X})$ with respect to the weak topology.
\end{lemma}

\begin{lemma}\label{lem:InfAttained}
The infimum in \eqref{funL} is attained.
\end{lemma}
\begin{proof}
We can assume $\lambda>1$, else the data processing inequality implies that $F(\lambda)\geq 0$, which is easily attained.
First, we show that if $\{U_n,X_n,Y_n,V_n\}_{n\geq 1}$ is a sequence of candidate minimizers\footnote{i.e., $(X_n,Y_n)$ equal $(X,Y)$ in distribution, and $U_n-X_n-Y_n-V_n$ for each $n$.} of \eqref{funL} which converge weakly to $(U^*,X^*,Y^*,V^*)$, then $(U^*,X^*,Y^*,V^*)$ is also a candidate minimizer.  To see this, note that Lemma \ref{lem:DSemiCont} asserts that relative entropy is lower semicontinuous with respect to the weak topology, and hence
\begin{align}
0 =\lim_{n\rightarrow \infty} D(P_{X_n Y_n} \| P_{XY}) \geq D(P_{X^* Y^*} \| P_{XY}) \geq 0.
\end{align}
Therefore, we see that $(X^*,Y^*)$ must be equal to $(X,Y)$ in distribution.  Similarly, by recognizing that (unconditional) mutual information is a relative entropy, lower semicontinuity also yields
\begin{align}
I(X;Y) =\lim_{n\rightarrow \infty} I(U_n,X_n;Y_n,V_n) \geq I(U^*,X^*;Y^*,V^*) \geq I(X^*;Y^*) = I(X;Y).
\end{align}
Hence, we can conclude that $U^*-X^*-Y^*-V^*$, and therefore $(U^*,X^*,Y^*,V^*)$ is a candidate minimizer as desired.

Suppose again that $\{U_n,X_n,Y_n,V_n\}_{n\geq 1}$ is a sequence of candidate minimizers that converges weakly to $(U^*,X^*,Y^*,V^*)$.  As  established previously, $(X_n,Y_n) = (X^*,Y^*) = (X,Y)$ in distribution.  Fix an arbitrary conditional distribution $P_{Z | Y}$.  Let $P_{Y} \rightarrow P_{Z | Y} \rightarrow P_{Z}$.
By lower semicontinuity of relative entropy, we have:
\begin{align}
\liminf_{n\rightarrow \infty} D( P_{X_n Y_n U_n} P_{Z} \| P_{X_n U_n} P_{Y Z}) \geq 
D( P_{X^* Y^* U^*} P_{Z} \| P_{X^* U^*} P_{Y Z})
\end{align}
Thus, there exists $\delta(n)\rightarrow 0$ such that 
\begin{align}
D( P_{X_n Y_n U_n} P_{Z} \| P_{X_n U_n} P_{Y Z}) \geq 
D( P_{X^* Y^* U^*} P_{Z} \| P_{X^* U^*} P_{Y Z}) -\delta(n),
\end{align}
which implies
\begin{align}
\inf_{P_{Z|Y}} D( P_{X_n Y_n U_n} P_{Z} \| P_{X_n U_n} P_{Y Z}) \geq 
\inf_{P_{Z|Y}} D( P_{X^* Y^* U^*} P_{Z} \| P_{X^* U^*} P_{Y Z}) -\delta(n).
\end{align}
By the variational representation of mutual information, the infima on the LHS and RHS are attained when $P_{Z|Y} = P_{U_n|Y}$ and $P_{Z|Y} = P_{U^*|Y}$, respectively.  Therefore,
\begin{align}
I(X_n ; Y_n | U_n) \geq  I(X^* ; Y^* | U^* ) -\delta(n).  
\end{align}
Since $\delta(n)$ vanishes, we can conclude that 
\begin{align}
\liminf_{n\rightarrow \infty}   I(X_n ; Y_n | U_n) &\geq  I(X^* ; Y^* | U^* ). \label{semi1}%
\end{align}

Now, we can add $2\lambda I(X;Y)$ to the functional being considered in \eqref{funL} without changing the nature of the optimization problem.  Therefore, we aim to show that the infimum of the functional 
\begin{align}
 &I(X;U)-\lambda I(Y;U) + I(Y;V|U)-\lambda I(X;V|U) + 2 \lambda I(X;Y)\\
 =&I(X;U) +  \lambda I(X;Y|U) + I(Y;V) + \lambda I(X;Y|V) + (\lambda-1) I(U;V).
\end{align}
is attained.  By our previous observation, if $\{U_n,X_n,Y_n,V_n\}_{n\geq 1}$ is a sequence of candidate minimizers which approach the infimum of \eqref{funL} and converge weakly to $(U^*,X^*,Y^*,V^*)$, then we can apply lower semicontinuity again together with \eqref{semi1} and its symmetric variant $\liminf_{n\rightarrow \infty}   I(X_n ; Y_n | V_n) \geq  I(X^* ; Y^* | V^* )$ to obtain
\begin{align}
F(\lambda) + 2 \lambda I(X;Y) &= \lim_{n\rightarrow \infty} I(X_n;U_n) +  \lambda I(X_n;Y_n|U_n) + I(Y_n;V_n) + \lambda I(X_n;Y_n|V_n) + (\lambda-1) I(U_n;V_n) \notag\\
&\geq I(X^*;U^*) +  \lambda I(X^*;Y^*|U^*) + I(Y^*;V^*) + \lambda I(X^*;Y^*|V^*) + (\lambda-1) I(U^*;V^*)\\
&\geq F(\lambda) + 2 \lambda I(X;Y),
\end{align}
implying equality throughout, and optimality of $(U^*,X^*,Y^*,V^*)$.

Thus, we only need to show that if  $\{U_n,X_n,Y_n,V_n\}_{n\geq 1}$ is a sequence of candidate minimizers, then there exists a subsequence $\{U_{n_k},X_{n_k},Y_{n_k},V_{n_k}\}_{k\geq 1}$ which converges weakly to some $(U^*,X^*,Y^*,V^*)$.  Since mutual information is invariant under one-to-one transformations of support, we can assume without loss of generality that $U_n,V_n$ are each supported on the interval $[0,1]$ for all $n$.  Recalling Prohorov's theorem \cite[Theorem 3.9.2]{durrett2010probability}, we only need to show that the sequence of measures $P_n$ is tight, where $P_n$ is the joint distribution of $(U_n,X_n,Y_n,V_n)$.  To this end, note that for any $\epsilon>0$, we can choose $t$ sufficiently large so that
\begin{align}
P_n \left\{ (U_n,X_n,Y_n,V_n) \notin [-t,t]^4\right\} = P_n \left\{ (X_n,Y_n) \notin [-t,t]^2 \right\} < \epsilon.
\end{align}
Thus, the claim is proved.
\end{proof}

\section{Auxiliary Lemmas}\label{app:Calculus}

\begin{lemma}\label{lem:hXE_hX}
For $n= 1, 2, \dots$, suppose $X^n \sim N(0,\Sigma_n)$, where $\Sigma_n \in \mathbb{R}^{n\times n}$ is positive definite, and let $E_n\in\{0,1\}$ be correlated with $X^n$.  If $\lim_{n\rightarrow \infty} \Pr\{E_n = 0\} = 1$, then
\begin{align}
\lim_{n\rightarrow \infty } \frac{1}{n}\Big( h(X^n|E_n=0) - h(X^n) \Big) = 0. \label{normalizeByn}
\end{align}
\end{lemma}
\begin{proof}
For the proof, we suppress the explicit dependence of $X^n$ and $\Sigma_n$ on $n$, and simply write $\vecX$ and $\Sigma$.
Note that 
\begin{align}
H(E_n)\geq I(\vecX;E_n) &= \Pr\{E_n=0\}D(P_{\vecX|E_n=0} \| P_{\vecX} ) + \Pr\{E_n=1\}D(P_{\vecX|E_n=1} \| P_{\vecX} ) \\
&\geq \Pr\{E_n=0\}D(P_{\vecX|E_n=0} \| P_{\vecX} ),
\end{align}
and therefore $D(P_{\vecX|E_n=0} \| P_{\vecX} )\rightarrow 0$.  Define $Z_n = \frac{1}{n}\vecX^T \Sigma^{-1} \vecX$, and $W_n = \frac{1}{n}\vecX^T \Sigma^{-1} \vecX |\{E_n=0\}$.  By the data processing theorem for relative entropy, $D(P_{W_n} \| P_{Z_n} )\rightarrow 0$ as well.  Thus, for any $\varepsilon>0$, Pinsker's inequality and the WLLN together imply
\begin{align}
\Pr\{ |W_n-1| \geq \varepsilon \} &= \Pr\{ |W_n-1| \geq \varepsilon \} - \Pr\{ |Z_n-1| \geq \varepsilon \} + \Pr\{ |Z_n-1| \geq \varepsilon \}\\
&\leq \| P_{W_n}-P_{Z_n} \|_{TV} + \Pr\{ |Z_n-1| \geq \varepsilon \}\\
&\leq \varepsilon
\end{align}
for all $n$ sufficiently large, establishing that $W_n\rightarrow 1$ in probability.  For any function $f : \mathbb{R}^n\rightarrow \mathbb{R}$%
\begin{align}
 \mathbb{E}\left[ \vecX^T \Sigma^{-1} \vecX  f(\vecX) \right] = \Pr\{E_n=0\} \mathbb{E}\left[ \vecX^T \Sigma^{-1} \vecX  f(\vecX) | E_n=0 \right]
 +  \Pr\{E_n=1\} \mathbb{E}\left[ \vecX^T \Sigma^{-1} \vecX  f(\vecX) | E_n=1 \right], \notag
\end{align}
so it follows that 
\begin{align}
\Pr\{E_n = 0\} \mathbb{E}\left[ W_n 1_{\{W_n \geq K\}}\right] \leq \mathbb{E}\left[ Z_n 1_{\{Z_n \geq K\}}\right] \label{unifIntegra}
\end{align}
by non-negativity of $W_n$ and $Z_n$.  The Cauchy-Schwarz and Markov inequalities together imply
\begin{align}
\left| \mathbb{E}\left[ Z_n 1_{\{Z_n \geq K\}}\right] \right|^2 \leq \mathbb{E}\left[ Z_n^2 \right] \Pr\{Z_n\geq K\} =\frac{n+2}{n} \Pr\{Z_n\geq K\}\leq \frac{3}{K}, 
\end{align}
and therefore \eqref{unifIntegra} implies that $\{W_n\}_{n\geq 1}$ is uniformly integrable.  It follows that $\mathbb{E}[W_n]\rightarrow 1 = \mathbb{E}[Z_n]$ (e.g., \cite[Theorem 5.5.2]{durrett2010probability}).

To conclude, we observe that
\begin{align}
\frac{1}{n}\Big( h(\vecX|E_n=0) - h(\vecX) \Big)
&= -\frac{1}{n}D(P_{\vecX|E_n=0} || P_\vecX ) + \frac{1}{n}\int_{\mathbb{R}^n} \Big( P_\vecX(\mathbf{x}) - P_{\vecX|E_n=0}(\mathbf{x}) \Big) \log(P_\vecX(\mathbf{x})) d\mathbf{x} \\
&= -\frac{1}{n}D(P_{\vecX|E_n=0} || P_\vecX ) -\frac{1}{2n} \mathbb{E}\left[ \vecX^T \Sigma^{-1} \vecX  \right] + \frac{1}{2n} \mathbb{E}\left[ \vecX^T \Sigma^{-1} \vecX | E_n=0 \right] \\
&= -\frac{1}{n}D(P_{\vecX|E_n=0} || P_\vecX ) +\frac{1}{2}\left(\mathbb{E}[W_n] - \mathbb{E}[Z_n]\right),
\end{align}
which completes the proof.
\end{proof}

Lemma \ref{lem:hXE_hX} can be viewed as a regularity property enjoyed by Gaussian vectors.  Though not needed elsewhere in this paper, it is interesting to note that Lemma \ref{lem:hXE_hX} is sharp in the following sense:  
\begin{proposition}
For $n= 1, 2, \dots$, let $X^n \sim N(0,\Sigma_n)$, where $\Sigma_n \in \mathbb{R}^{n\times n}$ is positive definite.  For any function $g(n)\rightarrow \infty$ arbitrarily slowly, there exists a sequence of random variables $E_n\in\{0,1\}$ such that $\lim_{n\rightarrow \infty} \Pr\{E_n = 0\} = 1$ and
\begin{align}
\lim_{n\rightarrow \infty } \frac{g(n)}{n}\Big( h(X^n|E_n=0) - h(X^n) \Big) = \infty.
\end{align}
In particular, the normalization by $1/n$ in \eqref{normalizeByn} is essential.
\end{proposition}
\begin{proof}
From the proof of Lemma \ref{lem:hXE_hX}, we can assume without  loss of generality that $X^n$ is iid $N(0,1)$.  Let $E_n = 0$ if there are at least $f(n)$ different $X_i$'s for which $|X_i| \geq \sqrt{2}$, and let $E_n = 1$ otherwise.

Since the $X_i$'s are independent, we can see that 
\begin{align}
\mathbb{E}\left[ \sum_i X_i^2 \Big| E_n = 0 \right] \geq 2 f(n) + (n-f(n)) = f(n) + n.  
\end{align}
On the other hand, 
\begin{align}
\mathbb{E}\left[ \sum_i X_i^2 \right] = n.
\end{align}

Now, we have that $\Pr\{E_n = 0\} = \Pr\{ B(n,p) \geq f(n) \}$, where $B(n,p)$ is a Binomial random variable consisting of $n$ trials with bias $p = \Pr\{ |X_i| \geq \sqrt{2} \}$.  Continuing, we have
\begin{align}
\Pr\{E_n = 0\} = \Pr\{ B(n,p) \geq f(n) \} = \Pr\left\{ \frac{1}{\sqrt{n}}(B(n,p) - np)\geq \sqrt{n}\left( \frac{1}{n} f(n) - p\right) \right\}.  
\end{align}
By the CLT,  $\frac{1}{\sqrt{n}}(B(n,p) - np) \rightarrow N(0,p(1-p))$ in distribution.  Therefore,  $\Pr\{E_n = 0\} \rightarrow 1$ provided $f(n) = o(n)$.  Recalling the proof of Lemma \ref{lem:hXE_hX}, we have
\begin{align}
h(X^n | E_n =0 ) - h(X^n)  \geq  \frac{1}{2}f(n) +o(1).
\end{align}
Thus, the claim is proved by putting $f(n) = n/\sqrt{g(n)}$, where $g(n)\rightarrow \infty$ arbitrarily slowly.
\end{proof}

\begin{lemma} \label{calculusLemma}
Consider a function $f : [0,\infty) \rightarrow \mathbb{R}$ defined implicitly by:
\begin{align}
2^{-2 t} = a_1 2^{-2 f(t)} + a_2. \label{implicit}
\end{align}
If $a_1 + a_2 \leq 1$, then
\begin{align}
\min_{t\geq 0} \Big\{ \max\{f(t),0\} - \lambda t \Big\}= \begin{cases}
\frac{1}{2}\log \left(\frac{a1 (\lambda-1)}{a2} \right) -\frac{\lambda}{2}\log \left(\frac{\lambda-1}{a_2 \lambda}\right) & \mbox{if $\lambda \geq \frac{a_1+a_2}{a_1}$}  \\
-\frac{\lambda}{2}\log\left(\frac{1}{a_1 + a_2}\right) & \mbox{if $0 \leq \lambda \leq \frac{a_1+a_2}{a_1}$}.
\end{cases}\label{dualCases}
\end{align}
\end{lemma}
\begin{proof}
Note that $f'(t) = \frac{2^{-2t}}{2^{-2t}-a_2}$, and therefore $f'(t) = \lambda \Rightarrow t = \frac{1}{2}\log \frac{\lambda-1}{a_2 \lambda}$.  Now, suppose $a_1 + a_2 \leq 1$.  Then $f(t) = 0 \Rightarrow t = \frac{1}{2}\log\frac{1}{a_1 + a_2}\geq 0$.  
Define $f_+(t) = \max\{f(t),0\}$.  Like $f(t)$, $f_+(t)$ is a convex increasing function.  For $\lambda > f'\left(\frac{1}{2}\log\frac{1}{a_1 + a_2}\right) = \frac{a_1+a_2}{a_1}$, $\lambda$ is a derivative of $f_+(t)$ at $t = \frac{1}{2}\log \frac{\lambda-1}{a_2 \lambda}$.  On the other hand, if $0 \leq \lambda \leq \frac{a_1+a_2}{a_1}$, then $\lambda$ is a subderivative of $f_+(t)$ at $t = \frac{1}{2}\log\frac{1}{a_1 + a_2}$.  Therefore, we can conclude that
\begin{align}
\min_{t\geq 0} \{ f_+(t) - \lambda t \}= \begin{cases}
\frac{1}{2}\log \left(\frac{a1 (\lambda-1)}{a2} \right) -\frac{\lambda}{2}\log\left( \frac{\lambda-1}{a_2 \lambda}\right) & \mbox{if $\lambda \geq \frac{a_1+a_2}{a_1}$}  \\
-\frac{\lambda}{2}\log\left( \frac{1}{a_1 + a_2} \right)& \mbox{if $0 \leq \lambda \leq \frac{a_1+a_2}{a_1}$}.
\end{cases}
\end{align}

\end{proof}
\vspace{-10pt}
\bibliographystyle{ieeetr}

\bibliography{myBib}

\end{document}

%% file: figs/Ellipse.pdf_tex
\begingroup%
  \makeatletter%
  \providecommand\color[2][]{%
    \errmessage{(Inkscape) Color is used for the text in Inkscape, but the package 'color.sty' is not loaded}%
    \renewcommand\color[2][]{}%
  }%
  \providecommand\transparent[1]{%
    \errmessage{(Inkscape) Transparency is used (non-zero) for the text in Inkscape, but the package 'transparent.sty' is not loaded}%
    \renewcommand\transparent[1]{}%
  }%
  \providecommand\rotatebox[2]{#2}%
  \ifx\svgwidth\undefined%
    \setlength{\unitlength}{951.5527145bp}%
    \ifx\svgscale\undefined%
      \relax%
    \else%
      \setlength{\unitlength}{\unitlength * \real{\svgscale}}%
    \fi%
  \else%
    \setlength{\unitlength}{\svgwidth}%
  \fi%
  \global\let\svgwidth\undefined%
  \global\let\svgscale\undefined%
  \makeatother%
  \begin{picture}(1,0.48297396)%
    \put(0,0){\includegraphics[width=\unitlength]{Ellipse.pdf}}%
    \put(0.01144363,0.46215914){\color[rgb]{0,0,0}\makebox(0,0)[lb]{\smash{$\mathbb{R}^n$}}}%
    \put(0.01303742,0.17787551){\color[rgb]{0,0,0}\makebox(0,0)[lb]{\smash{$\mathbb{R}^n$}}}%
    \put(0.06885662,0.26446779){\color[rgb]{0,0,0}\makebox(0,0)[lb]{\smash{$\vecX$-Encoder's Observations}}}%
    \put(0.06916324,0.20634507){\color[rgb]{0,0,0}\makebox(0,0)[lb]{\smash{$\vecY$-Encoder's Observations}}}%
    \put(0.05305771,0.44614258){\color[rgb]{0,0,0}\makebox(0,0)[lb]{\smash{$\vecX_3$}}}%
    \put(0.08037651,0.41037831){\color[rgb]{0,0,0}\makebox(0,0)[lb]{\smash{$\vecX_4$}}}%
    \put(0.03706391,0.36366863){\color[rgb]{0,0,0}\makebox(0,0)[lb]{\smash{$\vecX_1$}}}%
    \put(0.08377358,0.31356194){\color[rgb]{0,0,0}\makebox(0,0)[lb]{\smash{$\vecX_2$}}}%
    \put(0.03961171,0.07262622){\color[rgb]{0,0,0}\makebox(0,0)[lb]{\smash{$\vecY_1$}}}%
    \put(0.09264063,0.16519629){\color[rgb]{0,0,0}\makebox(0,0)[lb]{\smash{$\vecY_2$}}}%
    \put(0.15039078,0.12782856){\color[rgb]{0,0,0}\makebox(0,0)[lb]{\smash{$\vecY_3$}}}%
    \put(0.02130223,0.11848661){\color[rgb]{0,0,0}\makebox(0,0)[lb]{\smash{$\vecY_4$}}}%
    \put(0.32096438,0.27929655){\color[rgb]{0,0,0}\makebox(0,0)[lb]{\smash{Encoder}}}%
    \put(0.32075869,0.19200142){\color[rgb]{0,0,0}\makebox(0,0)[lb]{\smash{Encoder}}}%
    \put(0.91067396,0.14439564){\color[rgb]{0,0,0}\makebox(0,0)[lb]{\smash{$\mathbb{R}^n$}}}%
    \put(0.45440413,0.29264465){\color[rgb]{0,0,0}\makebox(0,0)[lb]{\smash{$knR_x$-bit}}}%
    \put(0.65062269,0.20000951){\color[rgb]{0,0,0}\makebox(0,0)[lb]{\smash{$\mathcal{E}(A_x,b_x)$}}}%
    \put(0.65056506,0.28563852){\color[rgb]{0,0,0}\makebox(0,0)[lb]{\smash{$\mathcal{E}(A_y,b_y)$}}}%
    \put(0.447638,0.26312952){\color[rgb]{0,0,0}\makebox(0,0)[lb]{\smash{Description}}}%
    \put(0.453587,0.20568404){\color[rgb]{0,0,0}\makebox(0,0)[lb]{\smash{$knR_y$-bit}}}%
    \put(0.44682087,0.1761689){\color[rgb]{0,0,0}\makebox(0,0)[lb]{\smash{Description}}}%
    \put(0.59536648,0.27201785){\color[rgb]{0,0,0}\makebox(0,0)[lt]{\begin{minipage}{0.19817235\unitlength}\raggedright \rotatebox[origin=c]{90}{Decoder}\end{minipage}}}%
  \end{picture}%
\endgroup%